\documentclass[journal,transmag]{IEEEtran}

\usepackage{amssymb}
\usepackage{amsthm}

\usepackage[utf8]{inputenc}\usepackage[T1]{fontenc}\usepackage{url}\usepackage{ifthen}\usepackage{cite}\usepackage[cmex10]{amsmath}
\usepackage{mathtools}

\usepackage{graphicx}\newtheorem{theorem}{Theorem}
\newtheorem{proposition}{Proposition}\newtheorem{lemma}{Lemma}\newtheorem{corollary}{Corollary}
\newtheorem{remark}{Remark}

\usepackage[normalem]{ulem}
\usepackage{multirow}
\usepackage{caption}
\usepackage{subcaption}
\usepackage{longtable}
\usepackage{courier}
\usepackage{float}
\usepackage{bbold}
\usepackage{dsfont}
\usepackage[verbose]{wrapfig}
\usepackage{amsthm}
\usepackage{authblk}
\usepackage{xcolor}

\newcommand{\dfn}{\stackrel{\triangle}{=}}

\def\dv{\boldsymbol{d}}

\def\Zc{\mathcal{Z}}
\def\Tau{\mathcal{T}}
\def\Tau{\mathcal{T}}

\usepackage[ruled,vlined,linesnumbered]{algorithm2e}

\newcommand{\defeq}{\triangleq}\newcommand{\Expect}{{\rm I\kern-.3em E}}

\newtheorem{example}{{\em Example}}

\newcommand\smallV{
  \mathchoice
    {{\scriptstyle\mathcal{V}}}
    {{\scriptstyle\mathcal{V}}}
    {{\scriptscriptstyle\mathcal{V}}}
    {\scalebox{.7}{$\scriptscriptstyle\mathcal{O}$}}
  }

\usepackage{color}

\begin{document}
\title{Full Coded Caching Gains for Cache-less Users}

\author{Eleftherios Lampiris and Petros Elia
\thanks{
E. Lampiris is with the Electrical Engineering and Computer Science Department, Technische Universit\"at Berlin, 10587 Berlin, Germany, email: lampiris@tu-berlin.de. The work was conducted while E. Lampiris was employed by EURECOM.

P. Elia is with the Communication Systems Department at EURECOM, Sophia Antipolis, 06410, France email: elia@eurecom.fr.

The work is supported by the European Research Council under the EU Horizon 2020 research and innovation program / ERC grant agreement no. 725929. (ERC project DUALITY).

Parts of this work were presented in ITW 2018 \cite{lampirisCachelessITW}.}
}
\maketitle

\begin{abstract}\nocite{maddah2014fundamental}
Within the context of coded caching, the work reveals the interesting connection between having multiple transmitters and having heterogeneity in the cache sizes of the receivers. Our work effectively shows that having multiple transmit antennas  -- while providing full multiplexing gains -- can also simultaneously completely remove the performance penalties that are typically associated to cache-size unevenness.
Focusing on the multiple-input single-output Broadcast Channel, the work first identifies the performance limits of the extreme case where cache-aided users coincide with users that do not have caches, and then expands the analysis to the case where both user groups are cache-aided but with heterogeneous cache-sizes. In the first case, the main contribution is a new algorithm that employs perfect matchings on a bipartite graph to offer full multiplexing as well as \emph{full} coded-caching gains to \emph{both} cache-aided as well as cache-less users.
An interesting conclusion is that, starting from a single-stream centralized coded caching setting with normalized cache size $\gamma$, then adding $L$ antennas allows for the addition of {up to} approximately $L/\gamma$ extra \emph{cache-less} users, at no added delay costs.
Similarly surprising is the finding that, {beginning} with a single-antenna hybrid system (with both cache-less and cache-aided users), then adding {$L-1$} antennas to the transmitter, as well as endowing the cache-less users with a cumulative normalized cache size $\Gamma_2$, increases the Degrees of Freedom by a \emph{multiplicative} factor of up to $\Gamma_{2}+L$.
\end{abstract}


\section{Introduction}

Coded caching is a technique --- first introduced in~\cite{maddah2014fundamental} for the single-stream bottleneck broadcast channel (BC) --- that exploits receiver-side caches in order to deliver cacheable content to many users at a time. This technique initially involved a setting where a single-antenna transmitter has access to a library of $N$ files, and serves (via a single bottleneck link) $K$ receivers, each having a cache of size equal to the size of $M$ files. The process involved a novel cache placement method and a subsequent delivery phase during which each user simultaneously requests one library file, while the transmitter employs cache-dependent coded multicasting to simultaneously deliver independently requested content to many users at a time.

In the single stream setting ($L=1$ transmit antenna), where the bottleneck link has capacity equal to $1$ file per unit of time, the work in~\cite{maddah2014fundamental} showed that any set of $K$ simultaneous requests can be served with \emph{normalized} delay (worst-case completion time, guaranteeing the delivery of any set of requested files) which is at most
\begin{align}
	T_{L=1}(K,\gamma) = \frac{K(1-\gamma)}{1+K\gamma},
\end{align}
where $\gamma \defeq \frac{M}{N} \in[0,1) $ denotes the normalized cache size.
This implied the ability to treat $K\gamma+1$ cache-aided users at a time; a number that is often referred to as the cache-aided sum \emph{Degrees of Freedom} (DoF)
\begin{align}
	D_{L=1}(K,\gamma) \defeq \frac{K(1-\gamma)}{T_{1}(K,\gamma)} = 1+K\gamma,
\end{align}
corresponding to a caching gain of $K\gamma$ additional served users due to caching.

\subsection{Multi-antenna coded caching}
Recently, coded caching has been explored in the presence of multiple antennas/transmitters. In the context of a fully-connected multiple-input single-output (MISO) BC, multi-antenna ($L$ antennas) techniques were combined with coded caching to reveal the new insight that multiplexing and caching gains can be combined additively \cite{shariatpanahiMultiServerTransIT,naderializadehFundamentalTransIT} to yield a sum-DoF of $D_{L}(K,\gamma)=L+K\gamma$, {which performance is shown in \cite{lampirisResolvingArXiv} to be \textit{exactly} optimal under the assumptions of uncoded placement and one-shot linear transmission schemes.}

{Another line of research has sought to ameliorate some of the bottlenecks appearing in coded caching through the use of multiple antennas/transmitters. Such results have shown that multiple antennas can i) dramatically reduce subpacketization, and in fact allow for multiplicative DoF gains in the finite file-size regime~\cite{lampirisSubpacketizationJSAC}, ii) achieve the full multi-antenna coded caching gains with feedback cost that is a function of the number of antennas and not of the caching gain~\cite{lampirisResolvingArXiv,lampirisCsitISIT,lampirisSubpacketizationCsitSPAWC}, iii) significantly improve the performance in the finite Signal-to-Noise-Ratio region \cite{ShariatpanahiPhysicalLayer2019TransIT,tolli2017multi,DMTlampiris2019Asilomar}, iv) augment the performance of coded caching through the knowledge of file popularity at the transmitter side~\cite{augmentingSerbetciWiOpt2020} as well as other insights (cf.\cite{bayat2018achieving,ngoScalableTransWireless2018,zhangTransIT2017,yiISIT2016tologogical,senguptaCISScache2016,caoTransWirelessFundamental2017,roigInterferenceICC2017,ZFE:15,zhangTopologicalISIT2017,tangResolvabilityISIT2016,lampirisNoCsitISIT,parrinello2018coded}, etc.).
 
The above synergistic nature of combining multiple antennas with coded caching also extends to ameliorating the effects of cache-size heterogeneity, as we show in this work.}

\subsection{Coded caching with heterogeneous cache sizes}\label{parHetCachesLiterature}

While the first works on coded caching focused more on the setting where users have identically sized caches, in reality many communication systems may include users with heterogeneous caching capabilities. It is expected, for example, that users with different types of devices, such as laptops and mobile phones which, naturally, have different storage constraints, will be simultaneously active. Moreover, different users may well decide to allocate different amounts of their available storage (or none of it) for content caching.

These uneven storage constraints can conceivably hamper the performance of coded caching systems. For example, imagining a system that needs to treat cache-aided and cache-less users, we can see that when users request different content, transmitting coded messages to cache-aided users can preclude the cache-less users from receiving any valuable information.

The expectation that users with uneven cache constraints might co-exist, has sparked a number of recent works that sought to ameliorate the effects of cache size unevenness \cite{senguptaLayeredAsilomar2016,amiriDecentralizedUnevenCachesTransComm2017,cao2018coded,ibrahimCentralizedHeterogeneousCachesWCNC2017,ibrahimD2DheterogeneousICC2018,daniel2017optimization,asadi2018centralized,ibrahimHeterogeneousTransComm2019}. For example, the work in \cite{senguptaLayeredAsilomar2016} adopts the approach of splitting the caches into multiple layers while caching at each layer according to the algorithm of \cite{maddah2014fundamental}, and adjusted to the users' cache sizes and to the size allocated to each cache layer. Further, the work in \cite{amiriDecentralizedUnevenCachesTransComm2017} considered the uneven cache-size scenario under decentralized placement\footnote{The main idea behind decentralized placement is to circumvent the fact that the identity of users needs to be known during the placement phase. Thus, storing of content at the users migrates from the deterministic placement introduced in \cite{maddah2014fundamental} to a random placement.}. Moreover, the work in \cite{ibrahimD2DheterogeneousICC2018} explores the Device-to-Device setting with heterogeneous caches, while \cite{cao2018coded} investigates the fundamental limits of the single-stream Coded Caching with $2$ users and uneven cache sizes. The works in \cite{daniel2017optimization,ibrahimHeterogeneousTransComm2019,asadi2018centralized} view the heterogeneous cache problem as a set of optimization constraints, where the size of each conceivable subfile is optimized in order to reduce the transmission time. Another idea can be found in \cite{ibrahimCodedPlacementAsilomar18} which adopts a coded placement approach in order to further increase the coded caching gains in the heterogeneous cache setting.

While for the homogeneous case we know that the scheme in~\cite{maddah2014fundamental} is optimal under the assumption of uncoded placement \cite{wanOptimalityITW2016} (and approximately optimal \cite{yuFactorOf2TransIT2019} under any placement scheme), the optimal performance of the heterogeneous cache-size setting is not known in general. What is known though, from all the above works, is that cache-size heterogeneity has always entailed performance penalties compared to the homogeneous case.

\subsection{Current setting, and brief summary of contributions}

In the current setting, we study the role of multiple antennas in tackling the penalties associated with having heterogeneous caches. Specifically, we first consider a system where a set of $K_{1}$ users are each assisted by caches of some normalized size $\gamma_{1}>0$, while the remaining $K_{2}=K-K_{1}$ users are cache-less ($\gamma_{2}=0$). In the single-antenna case we will show that, under the assumption of uncoded placement, the \textit{optimal strategy} is to treat each set separately, thus revealing that a single-stream system is severely penalized by the presence of cache-less users.

Motivated by the above, we will then shift our focus to the study of the multiple antenna case ($L$ antennas), where we will show that for a wide range of parameters, we are able to simultaneously treat both user types with DoF $D_{L}(K_{1},\gamma_{1},K_{2},\gamma_{2}=0)=K_{1}\gamma_{1}+L$ that will be equal to that of the corresponding homogeneous setting. Moreover, for the other case, where the DoF of the homogeneous system cannot be achieved, we will show that the DoF performance is $L$ times higher than in the single antenna case, and we will show that this performance is \emph{exactly optimal} under the assumption of uncoded cache placement.

We will then proceed to explore how the performance is boosted when now the $K_2 = K-K_1$ users of the second group are each endowed with a cache of normalized size $\gamma_{2}\in (0, \gamma_{1})$. First, we will prove that the total DoF of $L+K_{1}\gamma_{1}+K_{2}\gamma_{2}$ can be achieved for a broad range of parameters. Further, for the case when this gain cannot be achieved, we will show that the same performance boost experienced in the cache-less case when adding $L-1$ antennas (by a multiplicative factor of $L$) can be also achieved by adding caches to the cache-less group. Specifically, starting from the single antenna setting with the cache-aided and cache-less user sets, then adding a cumulative cache equal to $\Gamma_{2}$ at the cache-less group and $L-1$ antennas we can achieve a \textit{multiplicative} DoF boost by a factor of up to $\Gamma_{2}+L$. The above results reveal the power of multiple antennas in significantly or entirely removing the negative effects of cache-size unevenness, as well as the powerful effect that modest amounts of caching can have in uneven scenarios.

Again we stress that the above binary scenarios have particular practical pertinence. In the first scenario, the cache-less users may reflect users that employ legacy devices that do not support cache-aided decoding or that may wish to opt-out of dedicating their storage for caching (cf. \cite{paschos2016CommMag}).
The second scenario with two distinct cache sizes $\gamma_{1},\gamma_{2}$ ($\gamma_{1}>\gamma_{2}> 0$) reflects the expectation that users are split between those that have laptop devices that can store more information, and those with mobile devices which generally have more serious storage constraints.

\subsection*{System Model}
The goal of this work is to study the DoF performance of the $L$-antenna MISO BC\footnote{We note that, while here we focus on the MISO BC, the results can be readily extended to the wired multiserver setting of \cite{shariatpanahiMultiServerTransIT}, where a set of $L$ servers are connected to $K$ users and where the transmitted messages are linearly combined to form a full rank matrix between servers and users. The results can also be readily extended to the multiple transmitter interference setting where each of the $K_{T}$ transmitters stores fraction $\gamma_{T}$ of the content, such that $K_{T}\gamma_{T}=L$ (cf.~\cite{naderializadehFundamentalTransIT}). Thus, all our results can be trivially applied in the multiple-transmitter/multi-server settings.} with $K$ single-antenna receiving users, which are equipped with caches of heterogeneous sizes. Specifically, we focus on a system where a user belongs to one of two sets; the $K_{1}$ users of set $\mathcal{K}_{1}$ are endowed with caches of normalized size $\gamma_{1}\in(0,1)$, while each of the remaining $K_2=K-K_{1}$ users of set $\mathcal{K}_{2}$ have caches of normalized size $\gamma_{2}\in [0,\gamma_{1})$. Each user simultaneously asks for a single -- different -- file, from a library of $N\ge K$ files, thus the metric of interest is the worst-case delivery time. In order to satisfy the users' demands, the base station transmits an $L\times 1$ vector $\mathbf{x}$. Then the signal at each receiver $k$ takes the form
\begin{align*}
	y_{k}=\mathbf{h}^{H}_{k} \mathbf{x} + w_{k},~~  k\in\{1,2,..., K\}\triangleq[K]
\end{align*}
where $\mathbf{h}_{k}\in\mathbb{C}^{L\times1}$ denotes the channel between the transmitter and receiver $k$, where $\mathbf{x}$ satisfies the power constraint $\mathbb{E}\{\|\mathbf{x}\|^2\}=P$, and where $w_{k}\thicksim \mathbb{C}\mathcal{N}(0,1)$ corresponds to the noise observed at user $k$. We assume that each {user} has all {the} necessary channel-state information, and that for a given {Signal-to-Noise-Ratio} (SNR) {value}, each link has capacity of the form $\log(SNR)+o(log(SNR))$.

Our aim is to design, for the heterogeneous system at hand, a pre-fetching and delivery algorithm that minimizes the worst-case completion time
\begin{align}
	T_L(K_1,\gamma_{1},K_2,\gamma_{2})
\end{align}
corresponding to each user asking for a different file.

\section{Main results}
We begin with the case where users of set $\mathcal{K}_{2}$ are cache-less ($\gamma_{2}=0$), and then we will generalize by treating the case where $\gamma_2 \in (0, \gamma_{1})$. {We emphasize that we focus on the cases where $K_1\gamma_{1}, K_{2}\gamma_{2}$ are integers, while the non-integer cases can be achieved using memory sharing (cf. \cite{maddah2014fundamental}).}

\subsection{Coexistence of cache-aided and cache-less users \label{sec:Cacheless}}
We start with a result that exemplifies --- in the single stream case of $L=1$ --- the problem with having cache-aided users coexisting with cache-less users. We will use notation
\begin{align}
	T_{K_{i}} \defeq \frac{K_{i}(1-\gamma_{i})}{1+K_{i}\gamma_{i}}
\end{align}
to describe the delay needed to serve, in the single antenna setting, $K_i$ cache-aided users with caches of normalized size $\gamma_i$ (in the absence of any cache-less users) using $1$ transmit antenna, where this performance is exactly optimal under the assumption of uncoded cache placement.

\begin{proposition}\label{theSingleStream}
In a single-stream BC with $K_1$ cache-aided users equipped with caches of normalized size $\gamma_1$ and $K_2$ additional cache-less users, the optimal delay, under the assumption of uncoded placement, takes the form
\begin{align}
T_1(K_1,\gamma_1,K_2,\gamma_{2}=0)&=  \frac{K_1(1-\gamma_1)}{1+K_1\gamma_1}+ K_2.\label{eq:singleStream}
\end{align}
\end{proposition}
\begin{proof}
	{Achieving the above result is direct through the use of the algorithm of \cite{maddah2014fundamental} for the cache-aided users and then serving the cache-less users separately.
	The optimality part of the }proof is relegated to Appendix~\ref{proofSingleStreamPenalty}.
\end{proof}

The above reveals that in the single stream case, every time a single cache-less user is added, there is a delay penalty of an entire unit of time, thus {showing} that the two types of users {need to} be treated separately. If such separation were to be applied in the multi-antenna case, the achievable performance would be
\begin{equation}\label{eq:SeparatedL}
	T_{L}(K_{1},\gamma_1, K_{2},\gamma_{2}=0)=\frac{K_1(1-\gamma_1)}{L+K_1\gamma_1}+ \frac{K_2}{L}
\end{equation}
so the $K_{2}$ cache-less users would experience only a multiplexing gain of $L$, and would experience no caching gain, {while the DoF would be strictly less than $K_{1}\gamma_{1}+L$.}

We proceed with the main result of this work.

{
\begin{theorem}\label{thm:general}
	In the MISO BC with $L\ge1$ antennas, $K_{1}$ cache-aided users equipped with cache of fractional size $\gamma_1$, and $K_{2}$
	cache-less users, the achievable delivery time takes the form
	\begin{align}	\nonumber		
		T_{L}&(K_{1},\gamma_{1},K_{2},\gamma_{2}=0) =\\ \label{eqGeneralResult}
		&\begin{dcases}
			T_{K_{1}}+\frac{K_{2}-(L-1)T_{K_{1}}}{\min\{L, K_{2}\} } , & K_{2}  >(L-1)T_{K_{1}}
			\\
			\frac{K_2+K_1(1-\gamma_1)}{K_{1}\gamma_1+L}, & K_{2} \le(L-1)T_{K_{1}}.
		\end{dcases}
	\end{align}
The first case is within a multiplicative factor of $2$ from the optimal performance, while it is \textit{exactly} optimal under the assumptions of one-shot and linear schemes. Further,
the second case is within a multiplicative factor of $3$ from optimal under the assumption of linear and one-shot schemes.
\end{theorem}
\begin{proof}
	The achievability part of the proof is described in Section~\ref{sec:Description}. The outer bound {for the first case is detailed in Appendix~\ref{proofGapCacheless} (multiplicative factor of $2$) and Appendix~\ref{appendix2} (exact optimality under uncoded placement). The bound corresponding to the second case is detailed in Appendix~\ref{proofGapCacheless}.}
\end{proof}
}

{
\subsubsection*{Intuition}
The result of Theorem~\ref{thm:general} is divided in two cases relative to the number of cache-less users $K_{2}$. Beginning from the second case, which corresponds to complete mitigation of cache-size unevenness we can see that the number of cache-less users need to be less than $(L-1) T_{K_{1}}\approx \frac{L-1}{\gamma_{1}}$. A direct consequence of this scenario is that each antenna that is added to the system can serve up to $\frac{1}{\gamma_{1}}$ extra cache-less users without an increase on the delivery time (see Corollary~\ref{cor:Cor2}).

On the other hand, the first case of Theorem~\ref{thm:general} is achieved when the number of users is higher than the given threshold. As such, we can view this threshold as the limiting ability of a multi-antenna cache-aided system to serve cache-less users with the maximal DoF. It is interesting to note, as we also describe in the following corollary that, in this second case, the achieved performance is reduced by a factor $L$ compared to the single antenna setting.

}

\begin{corollary}\label{thmLfoldBoost}
Starting from the single-antenna BC with $K_1$ cache-aided users with caches of normalized size $\gamma_1$ and $K_{2}=(\tilde{L}-1)T_{K_{1}}$ cache-less users (for any positive $\tilde{L}$), then going from $1$ to $L\le \tilde{L}$ antennas, reduces delay by $L$ times, to a delay that is optimal under the assumption of uncoded cache placement.
\end{corollary}
\begin{proof}
The calculation of the performance comes directly from {Proposition}~\ref{theSingleStream} and {Theorem}~\ref{thm:general}.
\end{proof}

{
Corollary~\ref{thmLfoldBoost} shows the impact of equipping cache-aided networks with multiple antennas. While both resources provide significant reduction in the communication time, it is a joint utilization of both resources that can provide such an increase.

}

{Moreover}, we note that the above multiplicative boost in the DoF is in contrast to the additive DoF boost (additive multiplexing gain) experienced in systems with only cache-aided users~\cite{shariatpanahiMultiServerTransIT}.

Let us proceed with {some further} corollaries that explore some of the ramifications of the above theorem. Equation~\eqref{eq:SeparatedL} helps us place the following corollary into context.
\begin{remark} \label{remarkCor1}
In the $L$-antenna, $(K_{1},\gamma_{1},K_{2},\gamma_{2}=0)$ MISO BC with $K_2 \leq (L-1)T_{K_{1}}$, all cache-aided and cache-less users can experience full multiplexing gain $L$ as well as full caching gain $K_1\gamma_1$.
\end{remark}
{
Remark~\ref{remarkCor1} outlines the main difference between systems with one antenna compared to systems with multiple antennas. Furthermore, the remark brings forth the synergistic nature between multiple antennas and coded caching.

}

\begin{example}
	Let us assume the setting with $K_{2}=2$ cache-less users, and $K_{1}=5$ cache-aided users each equipped with a cache of normalized size $\gamma_{1}=\frac{1}{5}$. Transmitting {using} one antenna can achieve {the} optimal delay, under the assumption of uncoded placement, of
	\begin{align}
		T_{1}\left(5,\frac{1}{5}, 2,0\right)=\frac{K_{1}(1-\gamma_{1})}{K_{1}\gamma_{1}+1}+K_{2}=4.
	\end{align}
Going from $L=1$ to $L=2$ antennas, reduces the delay by a factor of 2, to the delay
	\begin{align}\label{eqDelay5_1_2_0}
		T_{2}\left(5,\frac{1}{5}, 2,0\right)=\frac{K_{1}(1-\gamma_{1})+K_{2}}{K_{1}\gamma_{1}+2}=2
	\end{align}
which is optimal under the assumption of uncoded cache placement.
\end{example}

We proceed with {a} corollary which can be placed into context, by noting that in a system with $L$ antennas and $K_2$ \emph{cache-less} users, adding one more antenna would allow (without added delay costs) the addition of only a diminishing number of $\frac{K_2}{L}$ extra cache-less users.

\begin{corollary} \label{cor:Cor2}
Let us start from the single-stream BC with $K_1$ cache-aided users equipped with caches of normalized size $\gamma_1$. Then, adding an extra $L-1$ transmit antennas, allows for the addition of
\begin{align}
	K_{2}=(L-1)T_{K_{1}}\approx \frac{L-1}{\gamma_{1}}
\end{align}
cache-less users, at no added delay costs.
\end{corollary}
\begin{proof} This is direct from Theorem \ref{thm:general}.
\end{proof}

The following takes another point of view and explores the benefits of injecting cache-aided users into legacy (cache-less) MISO BC systems. To put the following corollary into context, we recall that in a cache-less $L$ transmit-antenna MISO BC serving $K_2\geq L$ users, the optimal (normalized) delay is $\frac{K_{2}}{L}$.
\begin{corollary} \label{cor:StartMIMO}
In a MISO BC with $K_2\geq L$ cache-less users, introducing $K_1$ additional cache-aided users with $\gamma_1 \geq \frac{L}{K_2}$, incurs delay
\begin{align*}
	T_L(K_1,\gamma_1,K_2,\gamma_{2}=0) \leq \frac{K_2}{L-1}
\end{align*}
and thus we can add an infinite number of cache-aided users and only suffer a delay increase by a factor that is at most $\frac{L}{L-1}$.
\end{corollary}
\begin{proof} This is direct from Theorem~\ref{thm:general}.\end{proof}

\paragraph*{Multiple antennas for `balancing' cache-size unevenness}
In the variety of works (cf. Sec.~\ref{parHetCachesLiterature})
that explore the single-stream coded caching setting in the presence of uneven cache sizes, we see that having cache-size asymmetry induces delay penalties and that the preferred cache-size allocation is the uniform one. The following corollary addresses this issue, in the multi-antenna setting.

\begin{corollary} \label{cor:uneven}
The $L$-antenna, $(K_1,\gamma_{1},K_2,\gamma_{2}=0)$ MISO BC with $K_2\le (L-1)T_{K_{1}}$ cache-less users, incurs the same achievable delay
\begin{align*}
	T_{L}(K_1,\gamma_{1},K_2,\gamma_{2}=0) = \frac{K_2+K_1(1-\gamma_{1})}{L+K_1\gamma_{1}} = \frac{K(1-\gamma_{\text{av}})}{L+K\gamma_{\text{av}}}
\end{align*}
as the order optimal homogeneous $K$-user MISO BC with homogeneous caches of normalized size $\gamma_{\text{av}}\!=\! \frac{K_1\gamma_{1}}{K}$ (same cumulative cache size $K_1\gamma_{1} = K\gamma_{\text{av}}$).
\end{corollary}
\begin{proof} This is direct from Theorem~\ref{thm:general}.\end{proof}

\begin{example}
	Let us assume the $(K_{1}=5, \gamma_{1}=1/5, K_{2}=2,\gamma_{2}=0)$ MISO BC setting with $L=2$ antennas. The performance of this setting, as shown in a previous example (cf.~\eqref{eqDelay5_1_2_0}), is $T_{2}=2$.

	This matches the performance of the $L=2$ antenna homogeneous system with $K=7$ users and $\gamma=1/7$, whose delay is again (cf. \cite{shariatpanahiMultiServerTransIT,naderializadehFundamentalTransIT})
	\begin{align}
		T_{2}\left(7,\frac{1}{7}\right)=\frac{K(1-\gamma)}{L+K\gamma}=2.
	\end{align}

\end{example}

\subsection{Coexistence of users with different cache sizes}
We now proceed to lift the constraint of cache-less users and consider the more general scenario of $\gamma_2 \in(0,\gamma_{1})$.

{
\begin{theorem}\label{the2typesPositiveGammas}
	In the $L$-antenna $(K_{1},\gamma_{1},K_{2},\gamma_{2}>0)$ MISO BC, 
	the achievable delivery time takes the form
	\begin{align}\label{eq2typesPositiveGammasLarge}\nonumber
		&T_{L}(K_{1},\gamma_{1},K_{2},\gamma_{2})=\\
		&\begin{dcases}
			 \frac{K_{1}(1-\gamma_{1})+ K_{2}(1-\gamma_{2})}{L+K_{1}\gamma_{1}+K_{2}\gamma_{2}}, & T_{K_{1}}\ge \frac{K_{2}(1-\gamma_{2})}{L-1+K_{2}\gamma_{2}}\\
			 T_{K_{1}}+
			\frac{K_{2}(1-\gamma_{2})-(L-1+K_{2}\gamma_{2})\cdot T_{K_{1}}}{\min\{K_{2}, L+K_{2}\gamma_{2}\}}, & \text{else}.
		\end{dcases}	
	\end{align}
\end{theorem} }
\begin{proof}
	The proof is constructive and is detailed in Section~\ref{secCacheAidedScheme}.
\end{proof}

\begin{remark}
Theorems \ref{thm:general} and \ref{the2typesPositiveGammas} show how adding either one {additional antenna} or {equipping} the cache-less users {with caches}, {would result in} the same increase {of} the DoF. Most importantly, by either increasing the number of antennas or increasing the caches of the weaker users helps to decrease the penalty due to the system heterogeneity and allows to achieve the homogeneous performance. As an example, let us assume the $L$-antenna $(K_{1},\gamma_{1}, L\cdot T_{K_{1}},\gamma_{2}=0)$ MISO BC where the delay is given by~\eqref{eqGeneralResult} as
	\begin{align}
		T_{L}(K_{1},\gamma_{1},L\cdot T_{K_{1}},0)= T_{K_{1}}+ \frac{T_{K_{1}}}{L}=T_{K_{1}}\frac{L+1}{L}.
	\end{align}

We {want to} show how increasing either the number of antennas by $1$, or adding a small cache to each of the cache-less users such that $K_2\gamma_{2}=1$, results in the same DoF performance. First, increasing the number of antennas to $L+1$, corresponds to the case described by~\eqref{eqGeneralResult} where
\begin{align}
	T_{L+1}\left(K_{1},\gamma_{1},K_{2},0\right)= T_{K_{1}}
\end{align}
corresponding to a DoF of
	\begin{align}
		D_{L+1}(K_{1},\gamma_{1},L\cdot T_{K_{1}},0)= L+1+K_{1}\gamma_{1}
	\end{align}
	
Further, in the initial setting with $L$ transmit antennas, adding a small cache to each of the cache-less users such that $K_{2}\gamma_{2}=1$, we can easily see that the achieved performance corresponds to~\eqref{eq2typesPositiveGammasLarge}, thus
\begin{align}\nonumber
	T_{L}\left(K_{1},\gamma_{1},L\cdot T_{K_{1}}, \frac{1}{L\cdot T_{K_{1}}}\right)&= \frac{K_{1}(1-\gamma_{1})+ K_{2} -1}{L+K_{1}\gamma_{1}+1}
\end{align}
which corresponds to the cache-aided DoF of
\begin{align}
	D_{L}\left( K_{1}, \gamma_{1}, L\cdot T_{K_{1}}, \frac{1}{L\cdot T_{K_{1}}}\right) = L+K_{1}\gamma_{1}+1.
\end{align}

\end{remark}

From the above remark, we can see that adding antennas or small caches to the cache-less users allows for the full DoF to be achieved. In other words, we can see that the two resources work in tandem to boost the DoF.

Further, the above multiplicative gains can also be achieved in a setting with $K_{1}$ users equipped with caches of normalized size $\gamma_{1}$ which coexist with some $K_{2}=(\tilde{L}-1)T_{K_{1}},~\tilde{L}>1$ cache-less users, by increasing the transmit antennas and/or adding cumulative cache of size $\Gamma_{2}$ to the cache-less users. Specifically, adding these two resources to the system can raise the DoF by a multiplicative factor of $L+\Gamma_{2}\le \tilde{L}$. As we will see, this is a direct outcome of exploiting the multiple antennas as a means of spatially separating users and hence treating \emph{in the same transmission} both user types.

\begin{example}
	Let us assume a single antenna system with $K_{2}=10$ cache-less users and $K_{1}=7$ cache-aided users equipped with caches of normalized size $\gamma_{1}=\frac{1}{7}$. First, we will calculate the performance of the above setting and then we will proceed to add one more antenna, i.e. $L'=2$ and finally, we will add caches to the cache-less users.
	
The first setting's performance can be calculated by~\eqref{eq:singleStream} to be
	\begin{align}
		T_{1}\left(7, 1/7, 10,0\right)=\frac{7-1}{2}+10=13
	\end{align}
while the second setting's performance, given by~\eqref{eqGeneralResult}, is
	\begin{align}
		T_{2}\left(7, 1/7, 10,0\right)=\frac{7-1}{2}+\frac{7}{2}=\frac{13}{2}.
	\end{align}
Finally, the third setting's performance is given by~\eqref{eq2typesPositiveGammasLarge}
	\begin{align}
		T_{2}\left(7, 1/7, 10,1/10\right)=\frac{7-1}{2}+\frac{3}{3}=4.
	\end{align}
	
	From the above we can see that doubling the number of antennas will halve the system delay. Furthermore, if we also equip cache-less users with caches of cumulative size $\Gamma_{2}=1$, while having $L'=2$ antennas, we can see that the delay is reduced by more than a multiplicative factor of $3$, compared to the original setting, which amounts to a multiplicative DoF boost of $L'+\Gamma_{2}=3$ and a further additive reduction attributed to the local caching gain.
\end{example}

\section{Scheme Description}\label{sec:Description}

We begin with the scheme description for the case where the cache-aided users co-exist with the cache-less users. This scheme will then serve as the basis for the scheme for the case where both user types have non-trivial cache sizes.

In both of these cases, the challenge of the algorithm lies in properly combining the delivery of content towards each of the two types of users, such that subfiles intended for one set are either ``cacheable'' or can be ``nulled-out'' via Zero-Force (ZF) precoding.

\paragraph*{Notation}

In describing the scheme, we will use the following notation. The file requested by user $k\in[K]$ will be denoted by $W^{d_{k}}$. Symbol $\oplus$ denotes the bit-wise XOR operator, $\mathbb{N}$ the set of natural numbers, and for $n,k\in\mathbb{N}$ we denote with $\binom{n}{k}$ the binomial coefficient. For set $A$ we denote its cardinality with $|A|$, while for sets $A,B$ we will use $A\setminus B$ to denote the difference set.
We denote with $\mathcal{H}^{-1}_{\lambda}$ the normalized inverse ($L\times L$ precoder matrix) of the channel matrix between the $L$ antennas and the $L$ users of some set $\lambda\subset[K]$, $|\lambda|=L$, where the rows of the precoder matrix are denoted by $\big\{\mathbf{h}^{\perp}_{\lambda\setminus\{l\}}\big\}_{l\in\lambda}$ and have the following property 
\begin{align}
	\mathbf{h}^{H}_{k}\cdot \mathbf{h}^{\perp}_{\lambda\setminus\{l\}}=
	\begin{cases}
		1,	& \text{ if } k=l\\
		0,	& \text{ if } k\in \lambda\setminus \{l\}\\
		\ne 0,	& \text{ if } k\in [K]\setminus \lambda.
	\end{cases}
\end{align}

We remind that for some $K_{i}= \{K_{1}, K_{2}\}$, we denote by
\begin{align}
	T_{K_{i}} \triangleq \frac{K_{i}(1-\gamma_{i})}{1+K_{i}\gamma_{i}},~~ i\in\{1,2\}
\end{align}
the delay required to treat only one set of users with a single transmit antenna.

Finally for sets $\chi, \beta \subset [K]$, we define XORs $X_{\chi}$ and $X_{\chi,\beta}$ as follows
\begin{align}\label{eqSimpleXOR}
	X_{\chi}&=\bigoplus_{k\in\chi} W^{d_{k}}_{\chi\setminus\{k\}}\\ \label{eqTwoTypeXOR}
	X_{\chi,\beta}&=\bigoplus_{k\in\chi} W^{d_{k}}_{\beta\cup\chi\setminus\{k\}}.
\end{align}

{
Before describing the algorithm, we present one instance of the transmission that allows the simultaneous serving of both the cache-aided and cache-less users. 

\begin{example}
Let us assume the $L=2$-antenna MISO BC, where $K_1=5$ users have caches of normalized size $\gamma_1 = \frac{1}{5}$, while $K_2=2$ users have no caches. For this setting Algorithm~\ref{alg:Delivery} allows us to simultaneously serve $K_{1}\gamma_{1}+1 =2$ cache-aided users and $L-1=1$ cache-less user, while we note that the example in its entirety is presented in Section~\ref{secCachelessExample}. 

A single transmission of Algorithm~\ref{alg:Delivery} aimed at serving users $1,2$ and $6$ takes the form
	\begin{equation}
		\mathbf{x}_{1,2}^{1}=\mathbf{h}^{\perp}_{6} A_{2}B_{1} + \mathbf{h}^{\perp}_{2} F_{1}
	\end{equation}
We can see that the above transmission serves cache-aided users $1,2$ through XOR $A_{2}B_{1}$ and at the same time serves cache-less user $6$ with subfile $F_{1}$. We note that for some file $W\in \{A, B,F\}$ its subfile $W_{k}$ is cached at user $k$, a process that follows closely from the algorithm of \cite{maddah2014fundamental} and which we present in detail in Section~\ref{secPlacementDeliveryCacheless}.
 
The received message $y_{k}$ at each of the users takes the form
\begin{align}
	y_{1} &=  \mathbf{h}_{1}^{H} \mathbf{h}^{\perp}_{6} A_{2}B_{1} +  \mathbf{h}_{1}^{H} \mathbf{h}^{\perp}_{2} F_{1}\\
	y_{2} &=  \mathbf{h}_{2}^{H} \mathbf{h}^{\perp}_{6} A_{2}B_{1} +  \mathbf{h}_{2}^{H} \mathbf{h}^{\perp}_{2} F_{1} = \mathbf{h}_{2}^{H} \mathbf{h}^{\perp}_{6} A_{2}B_{1}\\
	y_{6} &=  \mathbf{h}_{6}^{H} \mathbf{h}^{\perp}_{6} A_{2}B_{1} +  \mathbf{h}_{6}^{H} \mathbf{h}^{\perp}_{2} F_{1} = \mathbf{h}_{6}^{H} \mathbf{h}^{\perp}_{2} F_{1}
\end{align}
where we have ignored the noise for simplicity.

User $1$ is receiving a linear combination of all three messages, and using its cache can remove the unwanted subfiles and decode its desired subfile $A_{2}$. Further, users $2,6$ are assisted by precoding and receive only their desired messages. Specifically, User $2$ needs to also use its cache in order to decode subfile $B_{1}$ while User $6$ can directly decode its desired $F_{1}$.

\end{example}

}

\subsection{Placement and delivery in the presence of cache-less users}\label{secPlacementDeliveryCacheless}
We will first provide an overview of the algorithm, and will then proceed to describe this algorithm in detail.

\subsubsection{Overview of Algorithm}The first scheme that we will present is designed to serve demands of both cache-aided and cache-less users in the same transmission. 
\paragraph{Placement Phase}The placement phase follows the algorithm in \cite{maddah2014fundamental} and starts by breaking each file into
\begin{align}
	S=\binom{K_{1}}{K_{1}\gamma_{1}}
\end{align}
subfiles and then proceeds with storing each subfile at exactly $K_{1}\gamma_{1}$ receivers.

\paragraph{Transmission Design} A transmitted vector is built by first forming an information vector of length $L$. The first element of this information vector is a XOR, intended for some $K_{1}\gamma_{1}+1$ cache-aided users, which is designed as in the algorithm of \cite{maddah2014fundamental} (cf.~\eqref{eqSimpleXOR}), and thus can be decoded by all involved cache-aided users. Further, the remaining messages are $L-1$ uncoded subfiles, each intended for a different cache-less user. These subfiles are carefully picked to match the file indices associated with the XOR.

Due to the inability of a cache-less user to remove interference, we need to Zero-Force the $L-1$ unwanted messages (the XOR as well as the $L-2$ other uncoded messages) to that user. Thus the XOR is Zero-Forced away from the cache-less users, while each uncoded message is Zero-Forced away from $L-2$ cache-less users and away from one cache-aided user.

As suggested above, thus far we have used-up all the $L-1$ spatial degrees of freedom for Zero-Forcing the XOR, but we have $1$ remaining spatial degree of freedom for each of the uncoded messages. This spatial degree of freedom will be used to ZF all the uncoded messages away from a single cache-aided user. This allows all cache-less users to be able to receive interference-free their intended message, as well as allows one cache-aided user to receive without interference the XORed message, which it can naturally decode. Most importantly, this last part will allow for the other $K_{1}\gamma_{1}$ cache-aided users to be able to have one subfile (per file) in common, which is crucial as we will see.

Since there are no remaining spatial degrees of freedom to exploit, the residual $K_{1}\gamma_{1}$ cache-aided users will receive a linear combination of all the $L$ symbols, and will need to exclusively use their cached content to remove these interfering messages. To this end, we need to pick the indices of the uncoded messages in a way such that these messages are completely known to these users. Since a subfile can be cached at a maximum of $K_{1}\gamma_{1}$ users and, conversely, a set of $K_{1}\gamma_{1}$ cache-aided users can have only one subfile (per file) in common, it follows that the subfile index that is delivered in every cache-less user is the same and is defined exactly by the $K_{1}\gamma_{1}$ cache-aided users that we are treating.

\paragraph{Verification that all subfiles are transmitted} As we saw in the previous paragraph, each transmission is responsible for communicating one XOR for a set $\chi$ of $K_{1}\gamma_{1}+1$ cache-aided users, as well as $L-1$ subfiles intended for cache-less users, where these subfiles share the same index $\tau$.

As mentioned before, in order for all $K_{1}\gamma_{1}$ users to decode their intended subfile, it suffices to choose the subfile index such that $\tau\subset\chi$. The pairing between a XOR and a subfile index can be viewed as a perfect matching problem over a bipartite graph. In this bipartite graph, a node of the left-hand-side represents one of the XORs ($\chi$), while a node of the right-hand-side represents $L-1$ subfiles with the same index ($\tau$) that are desired by some $L-1$ cache-less users.

As we will show later on, this problem is guaranteed to have a solution in our case. While such a solution can be constructed numerically, in the following paragraph we will describe an algorithm that can provide an explicit perfect matching to our problem for any set of parameters. This perfect matching can be achieved by a slight increase of the subpacketization.

\subsubsection{Algorithm Details} In this section we describe the details of the algorithm, starting from the placement phase. Further, we continue with the delivery algorithm in a pseudo-code format (Alg.~\ref{alg:Delivery}) accompanied by its description.

\paragraph{Placement Phase}

Initially, each file $W^{n},~n\in[N]$, is divided into
\begin{align}\label{eqSubpacketizaionCacheless}
	S_{nc}=K_{1}(1-\gamma_{1})\binom{K_{1}}{K_{1}\gamma_{1}}
\end{align}
subfiles, where these subfiles are named according to
\begin{align*}
	W^{n}\to \{W_{\tau}^{n,\phi},~\tau\subset\mathcal{K}_{1},~|\tau |=K_{1}\gamma_{1}, ~\phi\in\mathcal{K}_{1}\setminus\tau\}.
\end{align*}
Then, cache $\mathcal{Z}_{k}$ of cache-aided user $k\in \mathcal{K}_{1}$ is filled according to
\begin{align}\label{eq:placement}
	\mathcal{Z}_{k}=&\{W^{n,\phi}_{\tau} : k\in\tau, \forall \phi\in\mathcal{K}_{1}\setminus\tau,~\forall n\in[N]\}.
\end{align}
This is identical to the original placement in~\cite{maddah2014fundamental}, and the extra subpacketization (corresponding to index $\phi$) will facilitate the aforementioned combinatorial problem of matching XORs with uncoded subfiles.

\paragraph{Delivery Phase}

We will first focus on the case of $K_{2}=(L-1)T_{K_{1}}$, where the delay
\begin{align}
		T_{L}(K_{1},\gamma_{1},K_{2},\gamma_{2}=0)= \frac{K_{1}(1-\gamma_{1})+ K_{2}}{L+K_{1}\gamma_{1}}
\end{align}

can be achieved by simultaneously treating $K_{1}\gamma_{1}+L$ users. The extension to an arbitrary number of cache-less users is based on the algorithm of the above case, and will be described later on.

\paragraph*{Matching Problem}

As we have argued, the demands of the cache-aided users are treated by default via each XOR $X_{\chi}$. At the same time, we are able to treat $L-1$ cache-less users, under the condition that their received subfile index $\tau$ is the same and that $\tau\subset\chi$.

Thus, the challenge presented in the creation of a transmitted vector is to match a XOR index $\chi$ with a subfile index $\tau$ such that $\tau\subset \chi$ and at the end each $\chi$ is matched to a unique $\tau$, in the case where $T_{K_{1}}=1$ while if $T_{K_{1}}>1$, then $\chi$ needs to be matched to one of the $T_{K_{1}}$ different\footnote{We note here that, as discussed before, this index $\tau$ is the common index of all subfiles meant for the $L-1$ cache-less users during this transmission.} $\tau$. This constitutes a perfect matching over a bipartite graph, where the left-hand-side (LHS) nodes represent the $\binom{K_{1}}{K_{1}\gamma_{1}+1}$ different $\chi$ indices, while a node of the right-hand-side (RHS) represents one of the $T_{K_{1}}$ copies of the $\binom{K_{1}}{K_{1}\gamma_{1}}$ different $\tau$ intended for some $L-1$ cache-less recipients.

This type of problem is guaranteed to have a solution when each node from the LHS is connected to exactly $d\in\mathbb{N}$ nodes of the RHS (see \cite{agnarsson2007graph}). In our problem, it is easy to see that each node of the LHS is connected to $d=T_{K_{1}}\cdot(K_{1}\gamma_{1}+1)=K_{1}(1-\gamma_{1})$ nodes of the RHS.

Since an algorithm that finds such a solution may have high complexity (for example see \cite{fukuda1994finding}) we, instead, present an algorithm that requires a slightly higher subpacketization, but can provide an instant solution to the above matching problem. Specifically, the subpacketization of~\eqref{eqSubpacketizaionCacheless} contains the term $K_{1}(1-\gamma_{1})$, thus creating $K_{1}(1-\gamma_{1})$ copies of each XOR $X_{\chi}$, while the same holds for every subfile $\tau$ intended for the cache-less users. Our algorithm achieves a perfect matching by matching node $(\phi,\tau)$ of the RHS, where $\phi\in\mathcal{K}_{1}\setminus\tau$, with one of the XORs $X_{\{\phi\}\cup\tau}$ of the LHS.

\begin{algorithm}[th!]\caption{Transmission in the Cache-less Case}\label{alg:Delivery}
Assume $T_{K_{1}}\in\mathbb{N}$.\\
Group cache-less users: \begin{align*}
	g_{1}&=\{K_{1}+1, K_{1}+2,...,K_{1}+L-1\},...,\\
	g_{T_{K_{1}}}&=\{K_{1}+(L-1)\cdot (T_{K_{1}}-1),...,K\}.
\end{align*}\\
 \For{all $\tau\subset \mathcal{K}_{1}, ~|\tau|=K_1\gamma_{1}$ (pick file index)}{
 \For{all $\phi\in\mathcal{K}_{1}\setminus \tau$  (pick precoded user)}{
Set $\chi=\tau\cup\{\phi\}$\\
\For{all $t\in \left[T_{K_{1}}\right]$ (pick cache-less group)}{
Transmit:
\begin{align*}
	\mathbf{x}_{\tau,\phi}^{t}\!=\!
	\mathcal{H}_{\{\phi\}\cup g_{t}}^{-1}
	\begin{bmatrix}
		X_{\chi}\\
		W^{d_{g_{t}(1)},\phi}_{\tau}\\
		\vdots\\
		W^{d_{g_{t}(L-1)},\phi}_{\tau}\\
	\end{bmatrix}
	.
\end{align*}
}
}
}
\end{algorithm}

\paragraph{Transmission}
The delivery phase, in the form of pseudo-code, is presented in Alg.~\ref{alg:Delivery}, which we describe in this paragraph. Transmission commences by splitting the cache-less users into $T_{K_{1}}$ sets with $L-1$ users each (Step 2). Then we pick set $\tau\subset \mathcal{K}_{1}$, $|\tau|=K_{1}\gamma_{1}$ (Step 3), where this set serves two purposes. First, it identifies the cache-aided users that will not be assisted by precoding, and second, it identifies the subfile index that the selected cache-less users will receive. Next, cache-aided user $\phi$ is picked from the remaining set of cache-less users $\mathcal{K}_{1}\setminus\tau$ (Step 4). Then, set $g_{t}$, containing some $L-1$ cache-less users, is picked (Step 6).

The transmitted vector is created by calculating the precoder matrix $\mathcal{H}_{\{\phi\}\cup g_{t}}^{-1}$ such that it forms the normalized inverse of the channel matrix between the $L$-antenna transmitter and users of set $\{\phi\}\cup g_{t}$. The precoder matrix is multiplied by the information vector, which is comprised of XOR $X_{\chi}$ (intended for users $\tau\cup \{\phi\}$) and the $L-1$ uncoded subfiles that are all indexed by $\tau$ (intended for the cache-less users of set $g_{t}$) (Step 7).

\paragraph{Decodability}

In each transmitted vector, we can identify two sets of users, i) those that are assisted by precoding (set $\{\phi\}\cup g_{t}$) and ii) those that are not (set $\tau$). For the ``precoding-aided'' set, we can immediately recognize that due to the form of the precoder, these users will receive only their intended message. In the special case of the cache-aided user $\phi$, decoding the received XOR will also require the use of its cache.

Users belonging to the second set (set $\tau$) will be receiving a linear combination of all $L$ messages i.e.,
\begin{align}
	y_{k\in\tau} =\mathbf{h}^{H}_{k}\mathbf{h}^{\perp}_{g_{t}} X_{\chi}+  \sum_{i=1}^{L-1} \mathbf{h}^{H}_{k}\mathbf{h}^{\perp}_{\lambda\setminus\{g_{t}(i)\}} W^{d_{g_{t}(i)},\phi}_{\tau} + w_{k}
\end{align}
where $\lambda=\phi\cup g_{t}$. We can see that all the terms in the summation are cached at all users in set $\tau$, thus can be removed from the equation. What remains is XOR $X_{\chi}$, which can be decoded (this is direct from \cite{maddah2014fundamental}) by all members of set $\chi$.

\paragraph{Transmitting unique subfiles every time}

At this point the reader may have noticed that the secondary subfile index, associated with the subfile of the cache-aided users, is not identified in Alg.~\ref{alg:Delivery}. This is intentional, since every time we transmit subfile $W^{d_{k}}_{\tau}$, we pick a new upper index such that all such indices have been picked. We continue to show that the number of times a subfile is transmitted is exactly $K_{1}-K_{1}\gamma_{1}$.

\begin{proof}
	Let us assume we are interested in delivering $W^{d_{k}}_{\mu}$ to a user belonging to the set of cache-less users. We can see that subfile index $\mu$ uniquely defines Step $3$, i.e. $\tau=\mu$, while the user's number, $k$, uniquely defines Step $6$. The algorithm goes over all possible $\phi\in\mathcal{K}_{1}\setminus\mu$ (Step 4), thus at the end, different parts of subfile $W^{d_{k}}_{\mu}$ will be delivered exactly $K_{1}(1-\gamma_{1})$ times to cache-less user $k$.

Now we turn our focus to some cache-aided user $k$ and examine how many times this user will receive something from subfile $W^{d_{k}}_{\mu}$. We need to count the number of times something from this subfile is delivered when user $k$ is assisted by precoding, as well as the number of times parts of this same subfile are transmitted when user $k$ is not assisted by precoding.

When user $k$ is assisted by precoding, it follows that the remaining cache-aided users are uniquely defined by $\mu$, i.e. $\tau=\mu$. Thus, the user's number defines Step $3$ while the user's request defines Step $4$. Then, Algorithm~\ref{alg:Delivery} will iterate Step $6$ a total of $T_{K_{1}}$ times.

Further, let us look at the case when user $k$ is not assisted by precoding, which means that $\chi=\{k\}\cup \mu$. As we know, the set of precoded users $\tau$ satisfies $\tau\subset \chi$, while there are a total of $K_{1}\gamma_{1} +1$ different $\tau$ for a specific $\chi$. Since user $k$ is not assisted by precoding, it follows that $k\in\tau$, which further means that the number of possible and different $\tau$ is $K_{1}\gamma_{1}$. For each of these $\tau$, there is a unique $\phi$ (Step $4$) and for each pair $(\phi,\tau)$ Step $6$ is iterated a total of $T_{K_{1}}$ times.

In total, the number of times parts of subfile $W^{d_{k}}_{\mu}$ are transmitted -- when $k$ is a cache-aided user -- is equal to
\begin{align}
	T_{K_{1}}+K_{1}\gamma_{1} \cdot T_{K_{1}}= K_{1}(1-\gamma_{1})
\end{align}
which concludes the proof.
\end{proof}

\subsubsection{Scheme generalization and delay calculation}

In this section, we will generalize Alg.~\ref{alg:Delivery} to the case of $T_{K_{1}}\notin\mathbb{N}$ and we will further calculate the delay of the scheme.

\paragraph{Scheme generalization} We remind that $(L-1)\cdot T_{K_{1}}$ represents the threshold beyond which we cannot serve all cache-less users with the maximum DoF. In the case where $(L-1)\cdot T_{K_{1}}\notin\mathbb{N}$, it follows that the number of cache-less users $K_{2}$ that we can serve using the maximum DoF must be either smaller or higher than $(L-1)\cdot T_{K_{1}}$, with both of these cases being treated in the following paragraph. If, on the other hand, $(L-1)\cdot T_{K_{1}}\in\mathbb{N}$, while $T_{K_{1}}\notin\mathbb{N}$, then we can simply increase the subpacketization by a multiplicative factor of $L-1$. This will create a bipartite graph with $(L-1)K_{1}(1-\gamma_{1})\binom{K_{1}}{K_{1}\gamma_{1}+1}$ nodes on the LHS and $(L-1)(K_{1}\gamma_{1}+1)\binom{K_{1}}{K_{1}\gamma_{1}}$ nodes on the RHS, which means that both numbers are integers, which means that the perfect matching can be achieved.

The other two constraints that we need to address, in order to generalize our algorithm, are the cases where $K_{2}\gtrless (L-1) T_{K_{1}}$, also corresponding to the case where $(L-1)\cdot T_{K_{1}}\notin\mathbb{N}$.

First, if $K_{2}< (L-1) T_{K_{1}}$, we proceed as in Alg.~\ref{alg:Delivery} but when the demands of the cache-less users have been completely satisfied, then we move to treat only the cache-aided users (through any multi-antenna algorithm, such as \cite{shariatpanahiMultiServerTransIT,naderializadehFundamentalTransIT,lampirisSubpacketizationJSAC,lampirisCsitISIT}), at a rate of $D_{L}(K_{1},\gamma_{1})=L+K_{1}\gamma_{1}$ users at a time, thus yielding the overall DoF of $D_{L}(K_{1},\gamma_{1},K_{2},\gamma_{2} = 0)=L+K_{1}\gamma_{1}$ for the whole duration of the transmission.

Finally, for the case of $K_{2}>(L-1)T_{K_{1}}$, delivery is split into two sub-phases. During the first sub-phase, we simply employ Alg.~\ref{alg:Delivery} on the first $(L-1)T_{K_{1}}$ cache-less users while simultaneously completing the delivery to all $K_1$ cache-aided users. This is done at a rate of $K_{1}\gamma_1+L$ users at a time. Then in the second sub-phase, we treat the remaining $K_{2}-(L-1)T_{K_{1}}$ cache-less users via ZF-precoding, $L$ users at a time. The above sums up to a total delay
\begin{equation*}
	T_{L}(K_{1},\gamma_{1},K_{2},\gamma_{2}=0)=T_{K_{1}}+\frac{K_{2}\!-\!(L-1)T_{K_{1}}}{\min\{L, K_{2}\}}.
\end{equation*}

\paragraph{Delay Calculation}

Following the steps of Alg.~\ref{alg:Delivery}, corresponding to the case of $K_{2}=(L-1)T_{K_{1}}$, we have
\begin{align}
	T_{L}\bigg(K_{1},\gamma_{1},&(L-1)\frac{K_{1}(1-\gamma_{1})}{K_{1}\gamma_{1}+1},0\bigg)
	=\\
	&\frac{\overbrace{\binom{K_{1}}{K_{1}\gamma_{1}}}^{\text{Step }3} \overbrace{K_{1}(1-\gamma_{1})}^{\text{Step }4}\overbrace{\frac{K_{1}(1-\gamma_{1})}{K_{1}\gamma_{1}+1}}^{\text{Step }6}}{ \underbrace{K_{1}(1-\gamma_{1})\binom{K_{1}}{K_{1}\gamma_{1}}}_{\text{Subpacketization}}}
	=\\
	& \frac{K_{1}(1-\gamma_{1})}{K_{1}\gamma_{1}+1} = \frac{K_{1}(1-\gamma_{1})+K_{2}}{K_{1}\gamma_{1}+L}.
\end{align}

\subsection{Scheme description for setting with heterogeneous cache-aided users} \label{secCacheAidedScheme}

In this section we consider the $L$-antenna MISO BC setting, where both user types are equipped with caches of heterogeneous sizes $\gamma_{1},\gamma_{2}\in (0,\gamma_{1})$. In the context of the single antenna heterogeneous setting, it has been an elusive goal to achieve the performance of the corresponding homogeneous system with $\gamma_{av}=\frac{K_{1}\gamma_{1}+K_{2}\gamma_{2}}{K}$ (i.e., of the homogeneous system with the same cumulative cache size constraint). What we will show here is that, for a wide range of parameters, the corresponding performance of the multi-antenna homogeneous setting can indeed be achieved in the multi-antenna heterogeneous setting.

\subsubsection{Algorithm overview}
First, we will focus on proving the result of~\eqref{eq2typesPositiveGammasLarge}, where we can see that each transmission serves exactly $L+K_{1}\gamma_{1}+K_{2}\gamma_{2}$ users.

The main idea is to use the extra spatial degrees of freedom as a way to separate some users that belong to one group from some users that belong to the other group.

As before, we will create an $L\times 1$ information vector, which will be multiplied by an $L\times L$ precoder matrix to form the transmitting vector. The elements of the created vector belong to one of $4$ types. One element corresponds to a XOR of $1+K_{1}\gamma_{1}$ subfiles intended for some users of set $\mathcal{K}_{1}$, while another element corresponds to a XOR of $1+K_{2}\gamma_{2}$ subfiles which is intended for some users of set $\mathcal{K}_{2}$. The remaining $L-2$ elements will carry $L_{1}-1, ~L_{1}\in [1,L-1]$ uncoded messages for users of set $\mathcal{K}_{1}$ and $L_{2}-1,~ L_{2}\in [1,L-1]$ uncoded messages for users of set $\mathcal{K}_{2}$, where $L_{1}+L_{2}=L$, and where the exact values of variables $L_1$ and $L_{2}$ are calculated by solving the equality
\begin{align}\label{eqStreamsCalculation}
	\frac{K_{1}(1-\gamma_{1})}{L_{1} + K_{1}\gamma_{1}} =\frac{K_{2}(1-\gamma_{2})}{L_{2} + K_{2}\gamma_{2}}
\end{align}
under the constraint that $L_{1}\ge 1$.

In other words, the above solution allocates $L_1$ streams to the cache-aided users and $L_2$ streams to the others. This observation will allow us to view the problem at hand as a concatenation of two multi-antenna problems. In what follows, we will make use of a new multi-antenna Coded Caching algorithm corresponding to the homogeneous setting and which we present in detail in Appendix~\ref{secNewMultiantennaScheme}. Further, we will assume that variables $L_{1}, L_{2}$ are integers, while we relegate the non-integer case to Appendix~\ref{appExtensionTwoTypesNonIntegerStreams}.

\subsubsection{Algorithm Details}

\paragraph{Placement} We split each file $W^{n},~n\in[N]$ into
\begin{align}
	S_{c}=(K_{1}\gamma_{1}+L_{1})\binom{K_{1}}{K_{1}\gamma_{1}}(K_{2}\gamma_{2}+L_{2})\binom{K_{2}}{K_{2}\gamma_{2}}
\end{align}
subfiles, where each subfile $W^{n,\phi_{1},\phi_{2}}_{\tau_{1},\tau_{2}}$ is characterized by $4$ indices, $\phi_{1}\in[K_{1}\gamma_{1}+L_{1}]$, $\tau_{1}\subset\mathcal{K}_{1},~|\tau_{1}|=K_{1}\gamma_{1}$ and $\phi_{2}\in[K_{2}\gamma_{2}+L_{2}]$, $\tau_{2}\subset\mathcal{K}_{2},~|\tau_{2}|=K_{2}\gamma_{2}$, where indices $\tau_{1}$ and $ \tau_{2}$ reveal which users have cached this subfile from sets $\mathcal{K}_{1}$ and $\mathcal{K}_{2}$, respectively, while indices $\phi_{1}, \phi_{2}$ will help, as previously in the cache-less case, with the combinatorial problem of matching subfile indices with XORs.

The caches of the users are filled as follows
\begin{align}\label{eqPlacementTwoTypes1}
	\mathcal{Z}_{k_{1}\in\mathcal{K}_{1}}& =\{ 	W^{n,\phi_{1},\phi_{2}}_{\tau_{1},\tau_{2}}~~: k_{1}\in\tau_{1}, \forall \tau_{2},\phi_{1},\phi_{2}			\}\\\label{eqPlacementTwoTypes2}
	\mathcal{Z}_{k_{2}\in\mathcal{K}_{2}} &=\{ 	W^{n,\phi_{1},\phi_{2}}_{\tau_{1},\tau_{2}}~~: k_{2}\in\tau_{2}, \forall \tau_{1},\phi_{1},\phi_{2}		\}
\end{align}
where it is easy to see that the above placement respects the cache-size constraint of each user.

\paragraph{Delivery Phase}

\begin{algorithm}[th!]\caption{Transmission Process for Multi-Antenna Heterogeneous Coded Caching }\label{algDelivery2types}
	\For{all $\chi_{1}\subseteq[K_{1}]$, $|\chi_{1}|=K_{1}\gamma_{1}+1$}{
		\For{all $s_{1}\in\chi_{1}$}{
			\For{all $\chi_{2}\subseteq[K_{2}]$, $|\chi_{2}|=K_{2}\gamma_{2}+1$}{
				\For{all $s_{2}\in\chi_{2}$ }{
					Set: $\tau_{1}=\chi_{1}\setminus\{s_{1}\}$\\
					$\tau_{2}=\chi_{2}\setminus\{s_{2}\}$\\
					$\lambda=\{s_1\}\cup\{s_2\}\cup\beta_{\tau_{1},s_{1}}\cup\beta_{\tau_{2},s_{2}}$.\\
					Transmit:
		\begin{align*}
		\mathbf{x}^{s_{1},\tau_{1}}_{s_{2},\tau_{2}}=
		\begingroup
			\renewcommand*{\arraystretch}{1.7}
			\mathcal{H}^{-1}_{\lambda}\cdot
			\begin{bmatrix}
				X_{\chi_{1},\tau_{2}}\\
				W^{d_{\beta_{\tau_1,s_1}(1)}}_{\tau_1,\tau_{2}}\\
				\vdots\\
				W^{d_{\beta_{\tau_1,s_1}(L_1-1)}}_{\tau_1,\tau_{2}}\\
				X_{\chi_{2},\tau_{1}}\\
				W^{d_{\beta_{\tau_2,s_2}(1)}}_{\tau_{1},\tau_2}\\
				\vdots\\
				W^{d_{\beta_{\tau_2,s_2}(L_2-1)}}_{\tau_{1},\tau_2}
			\end{bmatrix}
		\endgroup
		\end{align*}	
				}	
			}
		}
	}
\end{algorithm}

Algorithm~\ref{algDelivery2types} describes the delivery phase in the form of a pseudo-code. We begin by noting that symbol $\beta^{(i)}_{\tau,s}\subseteq[K_{i}]\setminus\tau$  denotes a set of $L_{i}-1$ elements, which are selected to be the elements following $s\in[K_{i}]\setminus\tau$. For example, assuming that $L_{i}=2$ and $[K_{i}]\setminus \tau =\{1,2,3,4\}$, then $\beta^{(i)}_{\tau,1} =\{2,3\}$ and $\beta^{(i)}_{\tau,3} =\{1,4\}$. In what follows we will refrain from using the upper index $i$ when describing set $\beta$, for the sake of simplicity.

As mentioned before, the algorithm works as a concatenation of two multi-antenna Coded Caching schemes. Specifically, it begins (Step $1$) by picking a set of $K_{1}\gamma_{1}+1$ users ($\chi_{1}\subset\mathcal{K}_{1}$) and then (Step $2$) by selecting one of those users ($s_{1}\in\chi_{1}$) to be the precoding-assisted user. The following two steps (Step $3$ and Step $4$) are responsible for picking a set of $K_{2}\gamma_{2}+1$ users ($\chi_{2}\subset\mathcal{K}_{2}$) and user $s_{2}\in\chi_{2}$, respectively.

Once these $K_{1}\gamma_{1}+K_{2}\gamma_{2}+2$ users have been selected, the algorithm proceeds with the calculation of sets $\tau_1\subset \mathcal{K}_1, \tau_2\subset \mathcal{K}_2$ which correspond to the set of users that will not be assisted by precoding, as well as proceeds to define the set $\lambda$ that contains the precoding-assisted users from both sets $\mathcal{K}_1$ and $\mathcal{K}_2$.

In the last step, the algorithm creates the transmitting vector. First, it calculates the normalized inverse of the channel between the $L$-antenna transmitter and the users of set $\lambda$. Then, it forms the information vector which is comprised of $L$ elements. Two of these elements are XORs, $X_{\chi_{1}, \tau_{2}}$ and $X_{\chi_{2},\tau_{1}}$, while the remaining are $L_{1}-1$ and $L_{2}-1$ uncoded messages are respectively intended for some users from set $\mathcal{K}_1$ and set $\mathcal{K}_2$. The transmitting vector is formed as a multiplication of the precoding matrix $\mathcal{H}^{-1}_{\lambda}$ with the information vector.

We can see that having selected the XOR for users in $\mathcal{K}_{1}$, along with the precoded user (Steps 1 and 2), then the algorithm goes over all possible combinations of XORs and their respective users corresponding to set $\mathcal{K}_{2}$. In the case of users of set $\mathcal{K}_{1}$, this allows to deliver all the index pairs $(\phi_{2},\tau_{2})$ that correspond to the other set of users.

\paragraph{Decoding Process}

The decoding process is similar to that of Alg.~\ref{algDeliveryXORplusUncoded}. For the users in set $\lambda$ i.e., the precoding-assisted users, we can see that they receive only one of the $L$ messages, thus they either decode using a ZF precoder (users in $\lambda\setminus\{s_{1}\}\setminus\{s_{2}\}$) or they use a ZF decoder and continue to decode their respective XOR by use of their cached content (users $s_{1}$ and $s_{2}$).

The remaining users ($\chi_{1}\cup \chi_{2}\setminus\{s_{1}\}\setminus\{s_{2}\}$) will receive a linear combination of all $L$ information messages, which they can decode using the acquired CSI and their stored content. For example, any user $k\in \tau_{1}$ will receive
\begin{align}\nonumber
	y_{k} =&\mathbf{h}^{H}_{k}\mathbf{h}^{\perp}_{\lambda\setminus\{s_{1}\}}X_{\chi_{1},\tau_{2}} +\mathbf{h}^{H}_{k}\sum_{i=1}^{L-1} \mathbf{h}^{\perp}_{\lambda\setminus\beta_{\tau_{1},s_{1}(i)}}W^{d_{\beta_{\tau_{1},s_{1}(i)}}}_{\tau_{1},\tau_{2}}\\ \label{eqDecodingTwoTypes}
	+&\mathbf{h}^{H}_{k}\mathbf{h}^{\perp}_{\lambda\setminus\{s_{2}\}}X_{\chi_{2},\tau_{1}}+\mathbf{h}^{H}_{k}\sum_{i=1}^{L-1} \mathbf{h}^{\perp}_{\lambda\setminus\beta_{\tau_{2},s_{2}(i)}}W^{d_{\beta_{\tau_{2},s_{2}(i)}}}_{\tau_{1},\tau_{2}}
\end{align}
where, excluding the first summand, all the other terms above are completely known to any receiver of set $\tau_{1}$, and thus can be removed. The remaining XOR $X_{\chi_{1},\tau_{2}}$ is decodable by any user in set $\tau_{1}$.

\subsection{Extension to the remaining cases}

In this section we will prove the result of~\eqref{eq2typesPositiveGammasLarge}, which corresponds to the case where the number of streams that should be allocated to the group with the higher cache size is less than one. In this case, we simply treat the users of set $\mathcal{K}_{1}$ using one stream, and allocate the remaining $L_{2}=L-1$ streams for the second set of users, $\mathcal{K}_{2}$.

At some point in the transmission, all the files requested by set $\mathcal{K}_{1}$ have been successfully communicated, while users of set $\mathcal{K}_{2}$ require more transmissions to completely receive their files. This is because we transmit at a rate of $K_{1}\gamma_{1}+1$ to users of set $\mathcal{K}_{1}$ and with rate of $L-1+K_{2}\gamma_{2}$ to users of set $\mathcal{K}_{2}$, where $\frac{K_{1}(1-\gamma_{1})}{1 + K_{1}\gamma_{1}} <\frac{K_{2}(1-\gamma_{2})}{L-1 + K_{2}\gamma_{2}}$.

To complete the transmission of files to the second set of users, we employ any of the multi-antenna cache-aided algorithms serving $L+K_{2}\gamma_{2}$ users at a time. Thus, the completion time corresponding to the two sets of transmissions, takes the form
\begin{align}\nonumber
	T_{L}&(K_{1},\gamma_{1},K_{2},\gamma_{2})=\\
	&\frac{K_{1}(1-\gamma_{1})}{K_{1}\gamma_{1}+1}+
	\frac{K_{2}(1-\gamma_{2})-T_{1}^{(1)}(L-1+K_{2}\gamma_{2})}{\min\{K_{2}, L+K_{2}\gamma_{2}\}}.
\end{align}

\section{Examples}

In this section we will display two examples, one for the cache-less users case, and one for the other case. Both of the examples will present the cases where the full DoF of $D_{L}(K_{1},\gamma_{1},K_{2},\gamma_{2})=L+K_{1}\gamma_{1}+K_{2}\gamma_{2}$ is achievable.

We will use the standard notation for user demands, where $A\triangleq W^{d_{1}}$, $B\triangleq W^{d_{2}}$, and so on. We will also omit the symbol $\oplus$ in the description of the XORs, so for example, we will write $A_{\tau_{1},\tau_{2}}B_{\tau_{1}',\tau_{2}'}$ instead of $A_{\tau_{1},\tau_{2}}\oplus B_{\tau_{1}',\tau_{2}'}$, etc.

\subsection{Cache-less users example ($\gamma_2=0$)}\label{secCachelessExample}
In this example, we will consider the $L=2$-antenna MISO BC, where $K_1=5$ users have caches of normalized size $\gamma_1 = \frac{1}{5}$, while $K_2=2$ users have no caches.

First, each file $W^{n},~n\in [N]$ is subpacketized into
\begin{align}
	W^{n}\to \{ W^{n,\phi}_{\tau},~\tau\subset\mathcal{K}_{1},~|\tau |=K_{1}\gamma_1, ~\phi\in\mathcal{K}_{1}\setminus\tau\}.
\end{align}

The caches of the users in set $\mathcal{K}_{1}$ are filled according to Eq.~\eqref{eq:placement} thus, for example, the cache of the first user contains
\begin{align*}
	\mathcal{Z}_{1}=\big\{&W^{n,2}_{1}, W^{n,3}_{1},W^{n,4}_{1},W^{n,5}_{1},
	~~	\forall n\in[N]\big\}.
\end{align*}

Before we describe the entire sequence of transmitted vectors, we focus briefly on a single vector and its decoding at each user.

\paragraph{Transmission and decoding for a specific set of users}
The goal is to treat $K_1\gamma_1+L = 3$ users in each time-slot. Let us look in detail at one transmitted vector, where we treat cache-aided users $ 1$ and $2$ together with cache-less user $6$. In this case, we transmit
	\begin{align}
		\mathbf{x}_{1,2}^{1}&=\mathcal{H}_{26}^{-1}
	\begin{bmatrix}
		A_{2}^{1}B_{1}^{2}\\
		F_{1}^{2}
	\end{bmatrix}.
	\end{align}
Let us examine the decoding process at the users. First, we can see that User $2$ will receive --- due to ZF precoding and the design of the precoding matrix $\mathcal{H}_{26}^{-1}$ --- only the XORed message $A_{2}^{1}B_{1}^{2}$, and can thus proceed to cache-out $A_{2}^{1}$ and decode the desired subfile $B_{1}^{2}$. User $6$ will receive, again due to precoding, only its respective desired message $F_{1}^{2}$.
Finally, User $1$ will receive a linear combination of $A_{2}^{1} B_{1}^{2}$ and $F_{1}^{2}$, as follows
\begin{align}
	y_{1}=\mathbf{h}_{1}^{H} \mathbf{h}^{\perp}_{6} A_{2}^{1} B_{1}^{2} + \mathbf{h}_{1}^{H} \mathbf{h}^{\perp}_{2} F_{1}^{2} +w_{1}.
\end{align}
 First, by caching out $F_{1}^{2}$, User $1$ can decode the XOR, and then again by accessing its cache, User 1 can remove $B_{1}^{2}$ to decode its desired message $A_{2}^{1}$.

\paragraph{Sequence of transmissions}
We now proceed with the entire sequence of the 40 transmissions. Given that each file is subpacketized into $K_1(1-\gamma_1)\binom{K_1}{K_1\gamma_1} = 5(1-\frac{1}{5})\binom{5}{1} = 20$ subpackets, the $40$ transmissions will correspond to the desired delay of 
\begin{align}
	T_{2}\left( 	5,\frac{1}{5},2,0		\right)=\frac{K_{2}+K_{1}(1-\gamma_1)}{L+K_{1}\gamma_1}=2.
\end{align}

The transmissions are:
\begin{align*}
	&\mathbf{x}_{1,2}^{1}=\mathcal{H}_{26}^{-1}
	\begin{bmatrix}
		A_{2}^{1}B_{1}^{2}\\
		F_{1}^{2}
	\end{bmatrix}
	,
	&\mathbf{x}_{1,2}^{2}=\mathcal{H}_{27}^{-1}
	\begin{bmatrix}
		A_{2}^{3}B_{1}^{3}\\
		G_{1}^{2}
	\end{bmatrix}
	\\
	&\mathbf{x}_{1,3}^{1}=\mathcal{H}_{36}^{-1}
	\begin{bmatrix}
		A_{3}^{1}C_{1}^{3}\\
		F_{1}^{3}
	\end{bmatrix}
	,
	&\mathbf{x}_{1,3}^{2}=\mathcal{H}_{37}^{-1}
	\begin{bmatrix}
		A_{3}^{2}C_{1}^{2}\\
		G_{1}^{3}
	\end{bmatrix}
	\\
	&\mathbf{x}_{1,4}^{1}=\mathcal{H}_{46}^{-1}
	\begin{bmatrix}
		A_{4}^{1}D_{1}^{4}\\
		F_{1}^{4}
	\end{bmatrix}
	,
	&\mathbf{x}_{1,4}^{2}=\mathcal{H}_{47}^{-1}
	\begin{bmatrix}
		A_{4}^{2}D_{1}^{2}\\
		G_{1}^{4}
	\end{bmatrix}
	\\
	&	\mathbf{x}_{1,5}^{1}=\mathcal{H}_{56}^{-1}
	\begin{bmatrix}
		A_{5}^{1}E_{1}^{5}\\
		F_{1}^{5}
	\end{bmatrix}
	,
	&\mathbf{x}_{1,5}^{2}=\mathcal{H}_{57}^{-1}
	\begin{bmatrix}
		A_{5}^{2}E_{1}^{2}\\
		G_{1}^{5}
	\end{bmatrix}
	\\
	&\mathbf{x}_{2,1}^{1}=\mathcal{H}_{16}^{-1}
	\begin{bmatrix}
		A_{2}^{4}B_{1}^{4}\\
		F_{2}^{1}
	\end{bmatrix}
	,
	&\mathbf{x}_{2,1}^{2}=\mathcal{H}_{17}^{-1}
	\begin{bmatrix}
		A_{2}^{5}B_{1}^{5}\\
		G_{2}^{1}
	\end{bmatrix}
	\\
	&	\mathbf{x}_{2,3}^{1}=\mathcal{H}_{36}^{-1}
	\begin{bmatrix}
		B_{3}^{2}C_{2}^{3}\\
		F_{2}^{3}
	\end{bmatrix}
	,
	&\mathbf{x}_{2,3}^{2}=\mathcal{H}_{37}^{-1}
	\begin{bmatrix}
		B_{3}^{1}C_{2}^{1}\\
		G_{2}^{3}
	\end{bmatrix}
	\\
	&\mathbf{x}_{2,4}^{1}=\mathcal{H}_{46}^{-1}
	\begin{bmatrix}
		B_{4}^{2}D_{2}^{4}\\
		F_{2}^{4}
	\end{bmatrix}
	,
	&\mathbf{x}_{2,4}^{2}=\mathcal{H}_{47}^{-1}
	\begin{bmatrix}
		B_{4}^{1}D_{2}^{1}\\
		G_{2}^{4}
	\end{bmatrix}		
	\\
	&\mathbf{x}_{2,5}^{1}=\mathcal{H}_{56}^{-1}
	\begin{bmatrix}
		B_{5}^{2}E_{2}^{5}\\
		F_{2}^{5}
	\end{bmatrix}
	,
	&\mathbf{x}_{2,5}^{2}=\mathcal{H}_{57}^{-1}
	\begin{bmatrix}
		B_{5}^{1}E_{2}^{1}\\
		G_{2}^{5}
	\end{bmatrix}	
	\\
	&\mathbf{x}_{3,1}^{1}=\mathcal{H}_{16}^{-1}
	\begin{bmatrix}
		A_{3}^{4}C_{1}^{4}\\
		F_{3}^{1}
	\end{bmatrix}
	,
	&\mathbf{x}_{3,1}^{2}=\mathcal{H}_{17}^{-1}
	\begin{bmatrix}
		A_{3}^{5}C_{1}^{5}\\
		G_{3}^{1}
	\end{bmatrix}
	\\
	&\mathbf{x}_{3,2}^{1}=\mathcal{H}_{26}^{-1}
	\begin{bmatrix}
		B_{3}^{4}C_{2}^{4}\\
		F_{3}^{2}
	\end{bmatrix}
	,
	&\mathbf{x}_{3,2}^{2}=\mathcal{H}_{27}^{-1}
	\begin{bmatrix}
		B_{3}^{5}C_{2}^{5}\\
		G_{3}^{2}
	\end{bmatrix}
	\\
	&\mathbf{x}_{3,4}^{1}=\mathcal{H}_{46}^{-1}
	\begin{bmatrix}
		C_{4}^{3}D_{3}^{4}\\
		F_{3}^{4}
	\end{bmatrix}
	,
	&\mathbf{x}_{3,4}^{2}=\mathcal{H}_{47}^{-1}
	\begin{bmatrix}
		C_{4}^{1}D_{3}^{1}\\
		G_{3}^{4}
	\end{bmatrix}	
	\\
	&
	\mathbf{x}_{3,5}^{1}=\mathcal{H}_{56}^{-1}
	\begin{bmatrix}
		C_{5}^{3}E_{3}^{5}\\
		F_{3}^{5}
	\end{bmatrix}
	,
	&\mathbf{x}_{3,5}^{2}=\mathcal{H}_{57}^{-1}
	\begin{bmatrix}
		C_{5}^{1}E_{3}^{1}\\
		G_{3}^{5}
	\end{bmatrix}
	\\
	&\mathbf{x}_{4,1}^{1}=\mathcal{H}_{16}^{-1}
	\begin{bmatrix}
		A_{4}^{3}D_{1}^{3}\\
		F_{4}^{1}
	\end{bmatrix}
	,
	&\mathbf{x}_{4,1}^{2}=\mathcal{H}_{17}^{-1}
	\begin{bmatrix}
		A_{4}^{5}D_{1}^{5}\\
		G_{4}^{1}
	\end{bmatrix}
	\\
	&\mathbf{x}_{4,2}^{1}=\mathcal{H}_{26}^{-1}
	\begin{bmatrix}
		B_{4}^{3}D_{2}^{3}\\
		F_{4}^{2}
	\end{bmatrix}
	,
	&\mathbf{x}_{4,2}^{2}=\mathcal{H}_{27}^{-1}
	\begin{bmatrix}
		B_{4}^{5}D_{2}^{5}\\
		G_{4}^{2}
	\end{bmatrix}
	\\
	&\mathbf{x}_{4,3}^{1}=\mathcal{H}_{36}^{-1}
	\begin{bmatrix}
		C_{4}^{2}D_{3}^{2}\\
		F_{4}^{3}
	\end{bmatrix}
	,
	&\mathbf{x}_{4,3}^{2}=\mathcal{H}_{37}^{-1}
	\begin{bmatrix}
		C_{4}^{5}D_{3}^{5}\\
		G_{4}^{3}
	\end{bmatrix}	
	\\
	&
	\mathbf{x}_{4,5}^{1}=\mathcal{H}_{56}^{-1}
	\begin{bmatrix}
		D_{5}^{4}E_{4}^{5}\\
		F_{4}^{5}
	\end{bmatrix}
	,
	&\mathbf{x}_{4,5}^{2}=\mathcal{H}_{57}^{-1}
	\begin{bmatrix}
		D_{5}^{1}E_{4}^{1}\\
		G_{4}^{5}
	\end{bmatrix}
	\\
	&\mathbf{x}_{5,1}^{1}=\mathcal{H}_{16}^{-1}
	\begin{bmatrix}
		A_{5}^{3}E_{1}^{3}\\
		F_{5}^{1}
	\end{bmatrix}
	,
	&\mathbf{x}_{5,1}^{2}=\mathcal{H}_{17}^{-1}
	\begin{bmatrix}
		A_{5}^{4}E_{1}^{4}\\
		G_{5}^{1}
	\end{bmatrix}		
	\\
	&\mathbf{x}_{5,2}^{1}=\mathcal{H}_{26}^{-1}
	\begin{bmatrix}
		B_{5}^{3}E_{2}^{3}\\
		F_{5}^{2}
	\end{bmatrix}
	,
	&\mathbf{x}_{5,2}^{2}=\mathcal{H}_{27}^{-1}
	\begin{bmatrix}
		B_{5}^{4}E_{2}^{4}\\
		G_{5}^{2}
	\end{bmatrix}	
	\\
	&\mathbf{x}_{5,3}^{1}=\mathcal{H}_{36}^{-1}
	\begin{bmatrix}
		C_{5}^{2}E_{3}^{2}\\
		F_{5}^{3}
	\end{bmatrix}
	,
	&\mathbf{x}_{5,3}^{2}=\mathcal{H}_{37}^{-1}
	\begin{bmatrix}
		C_{5}^{4}E_{3}^{4}\\
		G_{5}^{3}
	\end{bmatrix}		
	\\
	&\mathbf{x}_{5,4}^{1}=\mathcal{H}_{46}^{-1}
	\begin{bmatrix}
		D_{5}^{2}E_{4}^{2}\\
		F_{5}^{4}
	\end{bmatrix}
	,
	&\mathbf{x}_{5,4}^{2}=\mathcal{H}_{47}^{-1}
	\begin{bmatrix}
		D_{5}^{3}E_{4}^{3}\\
		G_{5}^{4}
	\end{bmatrix}.
\end{align*}
The $40$ slots, each of normalized duration
\begin{align*}
	t_{s}= \left[(K_{1}(1-\gamma_{1})\binom{K_{1}}{K_{1}\gamma_{1}}\right]^{-1} = \frac{1}{20}
\end{align*}
imply a delay $T_{2}\left(5,\frac{1}{5},2,0\right)=2$, which matches the delay
\begin{align*}
	T_{2}\left(7,\frac{1}{7} \right)= \frac{7-1}{2+1}=\frac{K(1-\gamma_{av})}{2+K\gamma_{av}}=2
\end{align*}
that would be needed in the homogeneous case where the $K=7$ users would have an identical $\gamma_{av} = \frac{1}{7}$ (same cumulative cache $K\gamma_{av}=1$).

\subsection{Two Type Cache-aided Example}

In this section we present an example that illustrates the mechanics of the two user case. Specifically, we will focus on the $L=3$-antenna MISO BC, where $K_{1}=5$ users of set $\mathcal{K}_{1}$ are equipped with caches of normalized size $\gamma_{1}=\frac{2}{5}$, while $K_{2}=4$ users of set $\mathcal{K}_{2}$ are equipped with caches of normalized size $\gamma_{2}=\frac{1}{4}$. For this setting, the number of streams (cf. Eq.~\eqref{eqStreamsCalculation}) should be divided as $L_{1}=1$ and $L_{2}=2$.

We begin by splitting each file into
\begin{align*}
	S=(K_{1}\gamma_{1}+L_{1})(K_{1}\gamma_{2}+L_{2})\binom{K_{1}}{K_{1}\gamma_{1}}\binom{K_2}{K_{2}\gamma_{2}}= 360
\end{align*}
subfiles, where subfile $W^{n, \phi_{1}, \phi_{2}}_{\tau_{1},\tau_{2}}$ has indices $\tau_{1}\subset[5],~|\tau_{1}|=2$, $\phi_{1}\in[3]$, $\tau_2\subset[4],~|\tau_{2}|=1$, $\phi_{2}\in[3]$.

\subsubsection{Placement Phase} This phase is carried out according to Eq.~\eqref{eqPlacementTwoTypes1}-\eqref{eqPlacementTwoTypes2} where, for example, the caches of users $1\in\mathcal{K}_{1}$ and $6\in\mathcal{K}_{2}$ are filled as
\begin{align*}
	\mathcal{Z}_{1}&=\{W^{n,\phi_{1},\phi_{2}}_{12,\tau_{2}}, W^{n,\phi_{1},\phi_{2}}_{13,\tau_{2}}, W^{n,\phi_{1},\phi_{2}}_{14,\tau_{2}}, W^{n,\phi_{1},\phi_{2}}_{15,\tau_{2}}, \forall \tau_{2},\phi_{1},\phi_{2}\}\\
	\mathcal{Z}_{6}&=\{W^{n,\phi_{1},\phi_{2}}_{\tau_{1},6}, \forall \tau_{1},\phi_{1},\phi_{2}\}.
\end{align*}

\subsubsection{Delivery Phase}

For notational simplicity, we abstain from using indices $\phi_{1},\phi_{2}$. Further, we will only present one iteration of the algorithmic steps 1-2,  that delivers the first XOR ($A_{23}B_{13}C_{12}$ intended for users $1,2,3$) of the user set $\mathcal{K}_{1}$, while it goes through all other steps (Steps 3-8).
\begin{align*}
	&\mathbf{x}_{1, 23}^{6,7}=\mathcal{H}^{-1}_{168}
	\begin{bmatrix}
		 A_{23,7} B_{13,7} C_{12,7}\\
		 F_{23,7} G_{23,6}\\
		 H_{23,7}
	\end{bmatrix}
	\\ 
	&\mathbf{x}_{1, 23}^{7,6}=\mathcal{H}^{-1}_{178}
	\begin{bmatrix}
		 A_{23,6} B_{13,6} C_{12,6}\\
		 F_{23,7} G_{23,6}\\
		 H_{23,6}
	\end{bmatrix}	
	\\
	&\mathbf{x}_{1, 23}^{6,8}=\mathcal{H}^{-1}_{167}
	\begin{bmatrix}
		 A_{23,8} B_{13,8} C_{12,8}\\
		 F_{23,8} H_{23,6}\\
		 G_{23,8}
	\end{bmatrix}
	\\ 
	&\mathbf{x}_{1, 23}^{8,6}=\mathcal{H}^{-1}_{187}
	\begin{bmatrix}
		 A_{23,6} B_{13,6} C_{12,6}\\
		 F_{23,8} H_{23,6}\\
		 G_{23,6}
	\end{bmatrix}	
	\\
	&\mathbf{x}_{1, 23}^{6,9}=\mathcal{H}^{-1}_{167}
	\begin{bmatrix}
		 A_{23,9} B_{13,9} C_{12,9}\\
		 F_{23,9} I_{23,6}\\
		 G_{23,9}
	\end{bmatrix}	
	\\ 
	&\mathbf{x}_{1, 23}^{9,6}=\mathcal{H}^{-1}_{197}
	\begin{bmatrix}
		 A_{23,6} B_{13,6} C_{12,6}\\
		 F_{23,9} I_{23,6}\\
		 G_{23,6}
	\end{bmatrix}		
	\\
	&\mathbf{x}_{1, 23}^{7,8}=\mathcal{H}^{-1}_{179}
	\begin{bmatrix}
		 A_{23,8} B_{13,8} C_{12,8}\\
		 G_{23,8} H_{23,7}\\
		 I_{23,8}
	\end{bmatrix}		
	\\ 
	&\mathbf{x}_{1, 23}^{8,7}=\mathcal{H}^{-1}_{189}
	\begin{bmatrix}
		 A_{23,7} B_{13,7} C_{12,7}\\
		 G_{23,8} H_{23,7}\\
		 I_{23,7}
	\end{bmatrix}		
	\\
	&\mathbf{x}_{1, 23}^{7,9}=\mathcal{H}^{-1}_{178}
	\begin{bmatrix}
		 A_{23,9} B_{13,9} C_{12,9}\\
		 G_{23,9} I_{23,7}\\
		 H_{23,9}
	\end{bmatrix}
	\\ 
	&\mathbf{x}_{1, 23}^{9,7}=\mathcal{H}^{-1}_{196}
	\begin{bmatrix}
		 A_{23,7} B_{13,7} C_{12,7}\\
		 G_{23,9} I_{23,7}\\
		 F_{23,7}
	\end{bmatrix}
	\\
	&\mathbf{x}_{1, 23}^{8,9}=\mathcal{H}^{-1}_{186}
	\begin{bmatrix}
		 A_{23,9} B_{13,9} C_{12,9}\\
		 H_{23,9} I_{23,8}\\
		 F_{23,9}
	\end{bmatrix}
	\\ 
	&\mathbf{x}_{1, 23}^{9,8}=\mathcal{H}^{-1}_{196}
	\begin{bmatrix}
		 A_{23,8} B_{13,8} C_{12,8}\\
		 H_{23,9} I_{23,8}\\
		 F_{23,8}
	\end{bmatrix}
	.
\end{align*}

\subsection{Decoding Process}

The decoding process follows the decoding steps of Alg.~\ref{algDeliveryXORplusUncoded}. First, the members of set $\lambda$, i.e. the precoding-assisted users, will receive one of the $L$ messages, which can decode using their cached content.

Further, the users of set $\tau_{1}\cup \tau_{2}$ will receive a linear combination of all $L$ messages, which can decode using the acquired CSIT and their cached content.

As an example, we will look at the decoding of transmitted message $\mathbf{x}_{1, 23}^{6,7}$ at any intended user. First, we can see that the precoded users are $1,6,8$ and these users will receive
\begin{align*}
	y_{k\in\{1,6,8\}} = \mathbf{h}_{k}^{H}\mathbf{h}^{\perp}_{\{1,6,8\}\setminus\{k\}}
	\begin{cases}
		A_{23,7} B_{13,7} C_{12,7}, & k=1\\
		 F_{23,7} G_{23,6}, & k=6\\
		 H_{23,7}, & k=8
	\end{cases}
\end{align*}
where naturally any of these users can decode its intended subfile.

Then, for the remaining users (users $2,3,7$) the received signal takes the form
\begin{align}\nonumber
	y&_{k\in\{2,3,7\}}=\mathbf{h}_{k}^{H}\mathbf{h}^{\perp}_{\{6,8\}} A_{23,7} B_{13,7} C_{12,7}\\
	 &+ \mathbf{h}_{k}^{H}\mathbf{h}^{\perp}_{\{1,8\}}F_{23,7} G_{23,6}
	+ \mathbf{h}_{k}^{H}\mathbf{h}^{\perp}_{\{1,6\}} H_{23,7} + w_{k}.
\end{align}
We can easily see that each of these users can decode its desired subfile by caching-out any other interfering message.

\section{Conclusion and Final Remarks}

An interesting outcome (Theorem \ref{thmLfoldBoost}) is the fact that despite having abundant side information at a sizable number $K_1$ of receivers, going from 1 to $L$ antennas gives an $L$-fold DoF boost. This comes in obvious contrast to the cache-aided multiple antenna setting with only cache-aided users \cite{shariatpanahiMultiServerTransIT,naderializadehFundamentalTransIT,lampirisSubpacketizationJSAC,lampirisCsitISIT}, where adding antennas increases additively and not multiplicatively the DoF.

Furthermore, we showed that adding antennas can ameliorate and even remove the effects of cache-size asymmetries. This can be important in practical scenarios where $\gamma_{1}$ is expected to be small, which would then allow cache-aided users to boost the DoF performance of a large number ($\approx (L-1)/\gamma_{1}$) of cache-less users.  Finally we have seen in Theorem~\ref{the2typesPositiveGammas} the powerful (multiplicative) DoF effect that modest amounts of caching can have in uneven cache-size scenarios.

\begin{remark}
	The delay of the single-antenna system $(K_{1},\gamma_{1})$ BC is the same as the $L$-antenna MISO BC with extra added $K_{2}=(L-1)T_{K_1}\approx \frac{L-1}{\gamma_{1}}$ users.
\end{remark}

The above remark says that for every new antenna we add to the system, we can also treat an additional, fixed number of approximately $\frac{1}{\gamma_{1}}$ cache-less users \emph{without} increasing the overall delay.

\subsection*{Intuition on the Cache-less user Design} The algorithm, which is either optimal or near optimal, manages to achieve full coding gains by eliminating the previously encountered penalties of cache-size unevenness. Key to this, was the careful use of antenna-aided user separation.

In the case where cache-aided and cache-less users coexist, the scheme employs this separation in two ways. First, the scheme protects the cache-less users from the XOR and from each other.
Secondly, and most importantly, separation allowed $K_{1}\gamma_{1}$ cache-aided users to be able to have one subfile per file in common. Even though the employed cache-aided users can cache-out interfering subfiles, nevertheless any collection of $K_{1}\gamma_{1}+1$ of such users do not have any common subfile index cached, since a subfile is cached at exactly $K_{1}\gamma_{1}$ users. This obstacle was surmounted by Zero-Forcing the messages intended for the cache-less users away from one cache-aided user. This allowed for the aforementioned ability for the cache-aided users (that are not protected by precoding) to share a common subfile index, and thus, by design, to cache out all the subfiles intended for the cache-less users.
\appendices

\section{Proof of Theorem~\ref{theSingleStream}}\label{proofSingleStreamPenalty}

Toward proving Theorem~\ref{theSingleStream}, we adapt the approach of \cite{wanOptimalityITW2016}, to lower bound the delay for the case where, out of the $K$ users, only $K_1$ users have a cache. The bound will then also prove tight for all
\begin{align}
	L\geq \frac{K_2}{T_{K_{1}}}-1 = \frac{K_2(1+K_1\gamma)}{K_1(1-\gamma)}-1
\end{align}
as we will see in Appendix~\ref{appendix2}. 

The proof (for $L=1$) tracks closely steps\footnote{We note in advance that a naive adaptation of the approach in \cite{wanOptimalityITW2016}, where we would simply account for a reduced sum-cache constraint $K_1 M$ corresponding to a redundancy $t=\frac{K_1 M}{N}$, would yield a loose bound; for example when $L=1$, this naive bound would be $T\geq\frac{K-t}{t+1}$ which would then translate to $T\geq \frac{K_1(1-\gamma)}{1+K_1\gamma} + \frac{K_2}{1+K_1\gamma}$ which is loose as the optimal delay will turn out to be $T = \frac{K_1(1-\gamma)}{1+K_1\gamma} + K_2$.} from~\cite{wanOptimalityITW2016} which --- for the case of $K_1 = K$ (where all users have caches) --- employed index coding to bound the performance of coded caching.
Some of these steps are sketched here for the sake of completeness. Particular care is taken here to properly construct the bound's counting arguments in a way that accounts for the fact that specific symmetries that are essential to the approach in \cite{wanOptimalityITW2016}, do not directly hold here, simply because the set $\mathcal{K}_{1} = [K_1]$ of users that enjoy side information is only a subset of the users that request files.

We will begin with lower bounding, first for the case of $L=1$, the delay $T(\dv,\chi)$ for any generic caching-delivery strategy $\chi$ and any demand vector $\dv \in \mathcal{D}_{wc} \defeq \{\dv:  \ d_i \neq d_j, ~i,j\in[K],~i\neq j\}$ whose $K$ entries are all different.  In the following, we use $\Zc_i$ to denote the cache of each user $i$, where naturally $\Zc_i = \emptyset$ for $i\in \mathcal{K}_2 \defeq [K] \setminus \mathcal{K}_1$.

\paragraph*{Distinct caching problems and their corresponding index coding equivalents \label{sec:CachingProblem}} 

We first follow closely the approach in \cite{wanOptimalityITW2016} to describe the association between index coding and our specific caching scenario here. As in \cite{wanOptimalityITW2016}, each caching problem (defined by a demand vector $\dv \in \mathcal{D}_{wc}$) is converted into an index coding problem, by having each requested file $W^{d_i}$ split into $2^{K_1}$ disjoint subfiles $W^{d_i}_\Tau,\Tau\in 2^{[K_1]}$, where $\Tau\subset[K_1]$ indicates the set of users that have $W^{d_i}_\Tau$ cached.
Since no subfile of the form $W^{d_i}_\Tau, \; ~\Tau \ni i$ is requested, the index coding problem here is defined by $$K_1 2^{K_1-1} + K_2 2^{K_1}$$ requested subfiles, which form the nodes of the side-information graph $\mathcal{G}=(\mathcal{V}_{\mathcal{G}},\mathcal{E}_{\mathcal{G}})$, where $\mathcal{V}_{\mathcal{G}}$ is the set of vertices (each vertex/node representing a different demanded subfile $W^{d_i}_\Tau, \Tau\not\ni i$) and $\mathcal{E}_{\mathcal{G}}$ is the set of direct edges of the graph. We recall that an edge from node $W^{d_i}_\Tau$ to $W^{d_{i'}}_{\Tau'}$ exists if and only if $i'\in\Tau$.

As in \cite{wanOptimalityITW2016}, this allows us to lower bound $T(\dv,\chi)$ by using the index-coding converse from~\cite{li2017cooperative} which says that for a given $\dv,\chi$ --- with corresponding side information graph $\mathcal{G}_{\dv}=(\mathcal{V}_{\mathcal{G}},\mathcal{E}_{\mathcal{G}})$ with $\mathcal{V}_{\mathcal{G}}$ vertices/nodes and $\mathcal{E}_{\mathcal{G}}$ edges --- the delay is bounded as
\begin{equation}\label{eq:indexbound}T\geq \sum_{\smallV \in \mathcal{V_{J}}}|\smallV|
\end{equation}
for every acyclic induced subgraph $\mathcal{J}$ of $\mathcal{G}_{\dv}$, where $\mathcal{V}_{\mathcal{J}}$ denotes the set of nodes of the subgraph $\mathcal{J}$, and where $|\smallV|$ is the size of the message/subfile/node $\smallV$.

The following describes the acyclic graphs, and also directly shows that these remain acyclic after they are enlarged to account for the content requested by the cache-less users.
In the following we will consider permutations $\sigma\in S_{K_1}$ from the symmetric group $S_{K_1}$, and, for a given demand vector $\dv$, we will use $\mathcal{A}_{\dv} \defeq \cup_{i \in [K] \setminus [K_1]} W^{d_i}$ to denote the union of all content in $\dv$ that is requested by the users in $[K] \setminus [K_1]$.

\begin{lemma}\label{lem:cons_acyclic}
For any $\dv$ and any $\sigma\in S_{K_1}$, and for an acyclic subgraph $\mathcal{J}_{\dv,\sigma}$ of $\mathcal{G}_{\dv}$, is designed here to consist of all subfiles $\{W^{d_{\sigma(i)}}_{\Tau},~\forall i\in[K_1], \forall \Tau\subseteq [K_1]\setminus \{\sigma(1),\sigma(2),\dots,\sigma(i)\} \}$, then the enlarged graph $\mathcal{J}_{\dv,\sigma} \cup \mathcal{A}_{\dv} $ is also acyclic.
\end{lemma}
\vspace{3pt}\emph{Proof.} The proof that the subgraph $\mathcal{J}_{\dv,\sigma}$ is acyclic is direct from \cite[Lemma 1]{wanOptimalityITW2016}. The proof that $\mathcal{J}_{\dv,\sigma} \cup \mathcal{A}_{\dv} $ is also an acyclic graph, i.e., that the addition (on the original $\mathcal{J}_{\dv,\sigma}$) of all the nodes corresponding to $\mathcal{A}_{\dv}$ does not induce any cycles, follows by first recalling that a directed edge from node $W^{d_i}_\Tau$ to $W^{d_{i'}}_{\Tau'}$ exists if and only if $i'\in\Tau$, which thus tells us that an edge cannot be drawn from any node representing content from~$\mathcal{A}_{\dv}$, because any cache-less user $i \in {K} \setminus [K_1]$ cannot belong to any such $\Tau$ simply because $\Tau \subset [K_1]$.\vspace{3pt} \qed

Given the acyclic subgraph $\mathcal{J}_{\dv,\sigma} \cup \mathcal{A}_{\dv} $, we combine Lemma~\ref{lem:cons_acyclic} with \eqref{eq:indexbound} to get
\begin{equation}
T(\dv,\chi)\geq T^{LB}(\sigma,\dv,\chi)
\end{equation}
where
\begin{align}
&T^{LB}(\sigma,\dv,\chi) \defeq \sum_{\smallV \in \mathcal{V_{\mathcal{J}_{\dv,\sigma}}}\cup \mathcal{A}_{\dv}}|\smallV| \nonumber\\
& =  \!\!\!\!\! \sum_{\Tau\subseteq [K_1]\setminus \{\sigma(1)\}}|W^{\boldsymbol{d_{\sigma(1)}}}_{\Tau}| + \!\!\!\! \sum_{\Tau\subseteq [K_1]\setminus \{\sigma(1),\sigma(2)\}}   \!\!\!\! \!\!\!\! |W^{\boldsymbol{d_{\sigma(2)}}}_{\Tau}|+\dots  
\nonumber\\
&+ \sum_{\Tau\subseteq [K_1]\setminus \{\sigma(1),\dots,\sigma(K_1)\}}|W^{\boldsymbol{d_{\sigma(K_1)}}}_{\Tau}|
+|\mathcal{A}_{\dv}| . \label{eq:TLB}
\end{align}
Then, as in \cite{wanOptimalityITW2016}, we average over worst-case demands to get
\begin{align}
T^*&\defeq \min_{\chi} \max_{\dv \in \mathcal{D}_{W_c}} T(\dv,\chi)
\nonumber \\ &
\geq \min_{\chi} \max_{\dv \in \mathcal{D}_{W_c}} \max_{\sigma \in S_{K_1}} T^{LB}(\sigma,\dv,\chi)\nonumber \\
&\geq \min_{\chi} \frac{1}{|\mathcal{D}_{W_c}|}\frac{1}{|S_{K_1}|}  \sum_{\sigma \in S_{K_1}} \sum_{\dv \in \mathcal{D}_{W_c}} T^{LB}(\sigma,\dv,\chi)
\nonumber \\ &
\geq \min_{\chi} \frac{1}{ P(N,K) K_1! }  \sum_{\sigma \in S_{K_1}} \sum_{\dv \in \mathcal{D}_{W_c}} T^{LB}(\sigma,\dv,\chi)\label{eq:alternativedefinitionofT}
\end{align}
where in the above we use $P(N,K) \defeq \frac{N!}{(N-K)!}$.

Rewriting the summation in \eqref{eq:alternativedefinitionofT}, we get
\begin{align}\nonumber
&\sum_{\sigma \in S_{K_1}} \sum_{\dv \in \mathcal{D}_{W_c}} T^{LB}(\sigma,\dv,\chi) =  \\
&\sum_{i=0}^{K_1}\sum_{n\in[N]}\sum_{\Tau\subseteq[K_1]:|\Tau|=i} \!\!\!\! |W^n_{\Tau}| \cdot \!\!\!\! \underbrace{\sum_{\sigma \in S_{K_1}} \sum_{\dv \in \mathcal{D}_{W_c}} \!\!\!\! \mathds{1}_{\mathcal{V}_{\mathcal{J}_{\dv,\sigma}}}(W^n_{\Tau})}_{\defeq Q_{i}(W^n_\Tau)} + |\mathcal{A}_{\dv}| \label{eq:longinequality} 
\end{align}
where $\mathcal{V}_{\mathcal{J}_{\dv,\sigma}}$ is the set of vertices in the acyclic component subgraph $\mathcal{J}_{\dv,\sigma}$ for a given $\dv,\sigma$ pair, and where $\mathds{1}_{\mathcal{V}_{\mathcal{J}_{\dv,\sigma}}}(W^n_{\Tau})$ denotes the indicator function which takes the value of 1 only if $W^n_{\Tau} \subset \mathcal{V}_{\mathcal{J}_{\dv,\sigma}}$, else it is set to zero.

\paragraph*{Counting arguments accounting for cache-less users \label{sec:countingArguments}}
Our aim is to count the number of times, $Q_{i}(W^n_\Tau)$, that any specific subfile $W^n_\Tau$ appears in the summation in~\eqref{eq:longinequality}. To do this, we draw from the counting arguments in \cite[Section VII-C]{parrinello2018coded} which derives $Q_{i}(W^n_\Tau)$ for the case where $K$ users share $\Lambda \leq K$ caches, where each cache $r\subset [\Lambda]$ serves $\Lambda_r$ users. Adapting these steps\footnote{The following expression could not have been derived, had we simply substituted $K$ for $K_1$, in the corresponding $Q_i$ expression in \cite{wanOptimalityITW2016}. Such a naive approach would have essentially corresponded to treating the cache-less and cache-aided cases separately, and would not have allowed us to guarantee, among other things, that both cache-less and cache-aided users request different files.} in \cite[Section VII-C]{parrinello2018coded} gives that
\begin{align}\nonumber
Q_{i} &= Q_{i}(W^n_\Tau)\dfn \sum_{\sigma \in S_{K_1}} \sum_{\dv \in \mathcal{D}_{W_c}} \mathds{1}_{\mathcal{V}_{\mathcal{J}_{\dv,\sigma}}}(W^n_{\Tau}) \\ \nonumber
=&{N-1 \choose K-1}\sum_{r=1}^{K_1}P(K_1-i-1,r-1)(K_1-r)! \times\\
& \times (K-1)!(K_1-1)! (K_1-i).
\label{eq:Qi}
\end{align}

Setting $x_i\dfn\sum_{n\in[N]}\sum_{\Tau\subseteq[K_1]:|\Tau|=i}|W^n_{\Tau}|$ and recalling that
\begin{equation}\label{eq:sumfiles}
N=\sum_{i=0}^{K_1}x_i=\sum_{i=0}^{K_1}\sum_{n\in[N]}\sum_{\Tau\subseteq[K_1]:|\Tau|=i}|W^n_{\Tau}|
\end{equation}
we combine \eqref{eq:alternativedefinitionofT}, \eqref{eq:longinequality} and \eqref{eq:Qi}, 
to get
\begin{equation}\label{eq:compacteq}
T \geq \sum_{i=0}^{K_1}\frac{Q_{i}}{P(N,K)K_1!}x_{i}.
\end{equation}

We now resume counting to calculate $\frac{Q_{i}}{\Lambda!P(N,K)}$ for each $i = 0,1,\dots,K_1$.
Applying~\eqref{eq:Qi}, we see that
\begin{align} \label{eq:finanumber}		\nonumber
&\frac{Q_{i}}{K_{1} ! P(N,K)} =\frac{(N-1)!(N-K)!}{(K-1)!(N-K)!K_{1}!N!}  \times
\\ &
\times \sum_{r=1}^{K_{1}}(K-1)!(K_{1}-i)P(K_{1}-i-1,r-1)(K_{1}-r)!\nonumber
\\
&=\frac{1}{K_{1}!N}\sum_{r=1}^{K_{1}}(K_{1} -i)P(K_{1} -i-1,r-1)(K_{1}-r)!\nonumber\\
&=\frac{1}{K_{1} !N}\sum_{r=1}^{K_{1}} \frac{(K_{1} -i)(K_{1} -i-1)!(K_{1} -r)!}{(K_{1} -i-r)!}\nonumber
\\ &
=\frac{1}{K_{1} !N}\sum_{r=1}^{K_{1}}\frac{(K_{1} -i)!(K_{1} -r)!}{(K_{1} -i-r)!}\nonumber\\
&=\frac{1}{N}\sum_{r=1}^{K_{1}}\frac{(K_{1} -i)!(K_{1} -r)!i!}{K_{1} !(K_{1} -i-r)!i!}\nonumber
\\&
=\frac{1}{N}\sum_{r=1}^{K_{1}}\frac{{K_{1} -r\choose i}}{{K_1 \choose i}}=\frac{{K_{1}\choose i+1}}{{K_1 \choose i}N}=\frac{K_1-i}{(i+1)N}.
\end{align}

Now substituting \eqref{eq:finanumber} into \eqref{eq:compacteq}, we get that
\begin{align}
T(\chi)&\geq \sum_{i=0}^{K_1}\frac{K_1-i}{(i+1)N} x_{i} + \frac{K_{1} !P(N,K)}{K_{1} !P(N,K)} \underbrace{|\mathcal{A}_{\dv}|}_{K_2} \label{eq:LBwithxi_2}
\end{align}
where the use of the fraction $\frac{K_{1} !P(N,K)}{K_{1} !P(N,K)} = 1$ is meant to remind us the number of times acyclic graphs corresponding to $\mathcal{A}_{\dv}$ were invoked in the summation in~\eqref{eq:longinequality}, and where we also note that the expression above follows from the fact that all $\dv\in \mathcal{D}_{wc}$ force $|\mathcal{A}_{\dv}|  = K_2 $.

\paragraph*{Optimization\label{sec:optimization}}
At this point we observe that the crucial constant $\frac{K_1-i}{(i+1)N}$ derived for the part of the subgraph corresponding to cache-aided users, matches exactly the number $\frac{K_1-i}{(i+1)N}$ derived in \cite{wanOptimalityITW2016} for the $K=K_1$ case where all users can have a cache. Consequently, under the same file-size constraint given in \eqref{eq:sumfiles}, and given the current cache-size constraint $\sum_{i=0}^{K_1}i \cdot x_{i}\leq  K_1 M $, the expression in~\eqref{eq:LBwithxi_2} serves as a lower bound on the delay of scheme $\chi$ whose cache placement implies the set $\{x_i\}$.

Then, following the exact minimization steps in \cite[Proof of Lemma 2]{yuFactorOf2TransIT2019}, we get
\begin{equation}\label{eq:optimization1}
T(\chi)\geq \frac{K_1(1-\gamma)}{1+K_1\gamma} + K_2
\end{equation}
for integer $K_1\gamma$, whereas for all other values of $K_1\gamma$, this is extended to its convex lower envelop.

This concludes lower bounding $\max_{\dv} \in \mathcal{D}_{wc} T(\dv,\chi)$, and thus --- given that the right hand side of \eqref{eq:optimization1} is independent of $\chi$ --- lower bounds the performance for any scheme $\chi$, hence concluding the proof of the converse for Theorem~\ref{theSingleStream} for the case of $L = 1$.

\section{Converse and gap to optimal for Theorem \ref{thm:general}}\label{proofGapCacheless}
Let us first consider the gap to optimal for the case of $K_{2} \geq (L-1)T_{K_{1}}$.

We have seen that when $K_{2}=\alpha(L-1)T_{K_{1}},~\alpha\geq 1$, the achievable delay in~\eqref{eqGeneralResult} takes the form
\begin{align}\label{eq:time1}
    T_{L}(K_{1},\gamma_{1},K_{2},\gamma_{2})&=T_{K_{1}}+\frac{K_{2}-(L-1)T_{K_{1}}}{L}\\
    &=T_{K_{1}}+(\alpha-1)\frac{L-1}{L}T_{K_{1}} \\
	&=\frac{ T_{K_{1}} }{L}(\alpha L-\alpha+1).
\end{align}
For a lower bound on the minimum possible delay, we use
	\begin{align}\label{eq:time2}
		T^{\star}=\frac{\min\{K_{2},L\}}{L}=\min\left\{1,\frac{\alpha(L-1) T_{K_{1}} }{L} \right\}
	\end{align}
corresponding to the optimal delay required to satisfy only the cache-less users.
A quick calculation of the ratio between \eqref{eq:time1} and \eqref{eq:time2}, bounds the gap as
\begin{align}
	G=\frac{\alpha L\!-\!\alpha+1}{\alpha L -\alpha}=1+\frac{1}{\alpha(L-1)}\le 2.
\end{align}
When $K_{2}<L$, then $\alpha(L-1)T_{K_{1}}<L$, which again gives
\begin{align}
	G=\frac{ T_{K_{1}} } {L}(\alpha L-\alpha+1)<\frac{T_{K_{1}}(\alpha L-\alpha-1)}{\alpha(L-1)T_{K_{1}} }\le2.
\end{align}

For the case of $K_{2}=\alpha (L-1)T_{K_{1}},~~\alpha\le 1$, the lower bound takes the form
\begin{equation} \label{eq:outerBoundSmall}
	T_{L}(K_{1},\gamma_{1},K_{2},\gamma_{2})\geq \max\left\{\frac{\min\{L,K_{2}\}}{L},~\frac{1}{2}\frac{K_{1}(1-\gamma_1)}{K_{1}\gamma_1+L}\right\}
\end{equation}
where the first term corresponds to the optimal performance of an `easier' system where all the cache-aided users are removed, and where the second term corresponds to an easier system where all cache-less users are removed, and where --- for this latter type of system --- we know from \cite{naderializadehFundamentalTransIT} that treating $K_1\gamma_1+L$ users at a time is at most a factor of 2 from optimal, under the assumptions of linear and one-shot schemes.
Combining \eqref{eq:outerBoundSmall} with the achievable
\begin{align}
	T_{L}(K_{1},\gamma_{1},K_{2},\gamma_{2})=\frac{K_1+K_1(1-\gamma_{1})}{K_{1}\gamma_{1}+L}
\end{align}
from \eqref{eqGeneralResult}, yields a gap of
\begin{align*}
	G=\frac{\frac{K_{2}+K_{1}(1-\gamma_{1})}{K_{1}\gamma_{1}+L}}{\max\left\{\frac{K_{2}}{L},\frac{1}{2}\frac{K_{1}(1-\gamma_{1})}{K_{1}\gamma_{1}+L}\right\}}.
\end{align*}
To bound this gap, note that if $\frac{K_{2}}{L}>\frac{1}{2}\frac{K_{1}(1-\gamma_{1})}{K_{1}\gamma_{1}+L}$ then we know from before that
		\begin{align}\nonumber
			G=\frac{\frac{K_{2}+K_{1}(1-\gamma_1)}{K_{1}\gamma_1+L}}{\frac{K_{2}}{L}}=\frac{L}{K_{1}\gamma_{1}}+\frac{\frac{K_{1}(1-\gamma_{1})}{K_{1}\gamma_1+L}}{\frac{K_{2}}{L}}
			\le 1+2.
		\end{align}

Similarly when $\frac{K_{2}}{L}<\frac{1}{2}\frac{K_{1}(1-\gamma_1)}{K_{1}\gamma_1+L}$, the gap is bounded as
		\begin{align*}
			G=\frac{\frac{K_{2}+K_{1}(1-\gamma_1)}{K_{1}\gamma_1+L}}{\frac{1}{2}\frac{K_{1}(1-\gamma_1)}{K_{1}\gamma_1+L}}=\frac{\frac{K_{2}}{K_{1}\gamma_1+L}}{\frac{1}{2}\frac{K_{1}(1-\gamma_1)}{K_{1}\gamma_1+L}}+2\le 3
		\end{align*}
where the last step considers that $\frac{K_{2}}{L}<\frac{1}{2}\frac{K_{1}(1-\gamma_{1})}{K_{1}\gamma_{1}+L}$.

This concludes the proof of Theorem~\ref{thm:general}. \ \ \ \ \ \ \ \ \ \ \ \ \ \ \ \ \ \ $\square$

\section{Proof of Theorem \ref{thmLfoldBoost}}\label{appendix2}

In the considered MISO BC\footnote{As mentioned above, this system shares the same fundamental properties with the Interference Channel with cache-aided transmitters and, also, with the wired multi-server setting, thus the following proof is applicable to those settings as well.} setting, where {${\ell_{1},\ell_{2} }\in \{0,1, ..., L\}$, while $\ell_1+\ell_2=L$,} correspond to the maximum number {of streams dedicated to the cache-aided and cache-less users respectively}, we have the following trivial bound, under the assumption of uncoded placement
\begin{align}
	D_{L}(K_{1},\gamma_{1},K_{2},\gamma_{2}) \le  {\ell_{1}}\cdot (K_{1}\gamma_{1}+1) +{ \ell_{2}}= L +{\ell_{1} }K_{1}\gamma_{1}
\end{align}

In other words, each stream can provide either one DoF to a cache-less user or sum-DoF of $K_{1}\gamma_{1}+1$ to some cache-aided users. Thus, the minimum transmission time, for a specific set of parameters can be calculated by optimizing variable ${\ell_{1}}$ as follows
\begin{align}\label{eqOptStreams}
	T^{\star} \ge \min_{ {\ell_{1}} \in(0,L)}\max \left\{  \frac{K_{1}(1-\gamma_{1})}{ {\ell_{1}}(1+K_{1}\gamma_{1})}, \frac{K_{2}}{L-{\ell_{1} }}	.	\right\}
\end{align}

It is obvious that~\eqref{eqOptStreams} is minimized when the two quantities are equal, since both are continuous and one is increasing and the other is decreasing. Thus the point that achieves $T^{\star}$ is for ${\ell_{1}}=\frac{L}{\tilde{L}}$, resulting in the optimal delivery time of
\begin{align}
	T^{\star}=T_{K_{1}} \frac{\tilde{L}}{L}
\end{align}
which is the delivery time in~\eqref{eqGeneralResult}.

\section{A new Cache-Aided Multi-Antenna Delivery Algorithm}\label{secNewMultiantennaScheme}

In this section we present a new multi-antenna coded caching algorithm. The presentation is done for the general $L$-antenna MISO BC channel with $K$ cache-aided users, where each is equipped with a cache of normalized size $\gamma\in(0,1)$.

The main idea behind the algorithm is to transmit, in each slot, $K\gamma+L$ subfiles, again in an information vector with $L$ messages. We achieve this by creating an $L$-length vector which is further multiplied by a ZF precoding matrix. The entries of the vector consist of one XOR, comprized of $K\gamma+1$ subfiles (created exactly as in the algorithm of \cite{maddah2014fundamental}), and $L-1$ uncoded subfiles. We continue with the placement and delivery phases.

\subsection{Cache Placement}

Initially, each file is divided into $S=\binom{K}{K\gamma}$ subpackets, which are further split into $K\gamma+L$ smaller packets. We will assume that $T_{1}(K,\gamma)=\frac{K(1-\gamma)}{1+K\gamma}$ is an integer, while extending the scheme to non-integer values requires a slightly increased subpacketization. Users' caches are filled according to
\begin{align}\nonumber
	\mathcal{Z}_{k\in[K]}=\big\{ W^{n,\phi}_{\tau}: &\tau\subset[K], |\tau|=K\gamma, k\in\tau,\\
						&\forall\phi\in[K\gamma+L], \forall n\in[N]\big\}.
\end{align}

The purpose of index $\phi$ is to guarantee the delivery of ``fresh'' information, a total of $K\gamma+L$ subfiles for each associated index $\tau$. We will refrain from using this index in the following algorithm, in order to keep the notation clearer, but we will show that data from each subfile is transmitted $K\gamma+L$ times, thus showing that each individual $\phi,\tau$ pair will indeed be transmitted.

\subsection{Delivery Phase}

In each delivery slot, as discussed above, we will create a vector of size $L\times1$, where one of its entries will be a XOR comprized of $K\gamma+1$ subfiles, while the remaining $L-1$ entries will be uncoded subfiles. Then, the vector will be multiplied by an $L\times L$ precoder matrix, which is calculated as the normalized inverse of the channel between the $L$-antenna transmitter and a subset of the $K\gamma+L$ users, namely one of the users of the XOR and the $L-1$ users that will be the recipients of the uncoded messages. The process is written in the form of a pseudo-code in Alg.~\ref{algDeliveryXORplusUncoded} and will be further described in the following paragraph. We remind that $\mathcal{H}^{-1}_{\lambda}$ denotes the normalized inverse of the channel matrix formed between the $L$ antenna transmitter and the users in set $\lambda$, while $\beta_{\tau,k}\subseteq[K]\setminus\tau$ is a set of $L-1$ elements, which are selected to be the elements following the element $k\in[K]\setminus\tau$.

\begin{algorithm}[th!]\caption{Delivery Phase}\label{algDeliveryXORplusUncoded}
\For{all $\chi\subseteq[K], |\chi|=K\gamma+1$ (pick XOR)}{
	\For{all $s\in\chi$ (pick precoded user)}{
		Set: $\tau=\chi\setminus\{s\}$\\
		Set: $\lambda=\{s\}\cup\beta_{\tau,s}$.\\
		Transmit:
		\begin{align}
		\mathbf{x}_{s,\tau}=
		\begingroup
			\renewcommand*{\arraystretch}{1.7}
			\mathcal{H}^{-1}_{\lambda}\cdot
			\begin{bmatrix}
				\bigoplus_{k\in\chi} W^{d_{k}}_{\chi\setminus\{k\}}\\
				W^{d_{\beta_{\tau,s}(1)}}_{\tau}\\
				\vdots\\
				W^{d_{\beta_{\tau,s}(L-1)}}_{\tau}
			\end{bmatrix}
		\endgroup
		\end{align}
	}
}
\end{algorithm}

\paragraph*{Details of Algorithm \ref{algDeliveryXORplusUncoded}} The algorithm begins by selecting a subset $\chi$ of the users of size $K\gamma+1$. For these users, the algorithm will form a XOR in the same way as does the algorithm of \cite{maddah2014fundamental}. Then, the algorithm selects one user, $s$,  from the users of set $\chi$, where this user will be helped by precoding. It is easy to see that the subfile index that this user will receive is $\chi\setminus\{s\}=\tau$.

Further, the remaining $L-1$ users that are scheduled to receive from the transmitted vector, need to be selected. These users are described by set $\beta_{\tau,s}$, which is calculated by finding the $L-1$ consecutive elements of set $[K]\setminus\tau$ after element $s$. For example, if $\chi=\{1,2,3\}$, $K=5$, $L=2$ and $s=1$, then $[K]\setminus\tau=\{1,4,5\}$ thus, $\beta_{\tau,s}=\{4\}$, as $4$ comes right after element $s=1$. The users of set $\tau\cup \{s\} \cup \beta_{\tau,s}$ are the $L+K\gamma$ users that will receive a subfile in this slot.

For the above selected users, the algorithm creates an $L\times 1$ vector, where one of the elements is a XOR designed for the users in set $\chi$, while the remaining elements correspond to subfiles indexed with $\tau$ and intended for the users in set $\beta_{\tau, s}$.

Further, the algorithm forms the precoder matrix $\mathcal{H}_{\lambda}^{-1}$ such that it is the normalized inverse of the channel matrix between the $L$-antenna transmitter and the users in $\lambda=\{s\}\cup \beta_{\tau,s}$. Finally, the transmitted vector is created by multiplying the precoder matrix with the vector containing the messages.

\paragraph{Decoding Process} We begin with the users of set $\lambda$ i.e, the ``precoding-assisted'' users. Due to the design of the precoder, we can see that these users will receive either the XORed message (user $s$) or each of the uncoded messages to the respective user i.e.,
\begin{align}
	y_{k\in\lambda}=\mathbf{h}^{H}_{k}\mathbf{x}_{s,\tau} =
	\begin{dcases}
		\oplus_{k\in\chi} W^{d_{k}}_{\chi\setminus\{k\}}, & k=s\\
		W^{d_{k}}_{\tau}, & \text{else}
	\end{dcases}
\end{align}
where for simplicity we have removed the noise. It is easy to see that users in set $\beta_{\tau,s}$ will be assisted by precoding, thus will only ``see'' the uncoded subfile that they want. Further, user $s$ will receive XOR $X_{\chi}$ which can proceed to decode using its cached content.

On the other hand, users in set $\tau$ will be receiving a linear combination of all $L$ messages, which will proceed to decode using both CSIT knowledge and their cached subfiles. The received message at some user $k\in\tau$ takes the form
\begin{align}
	y_{k\in\tau}&=\mathbf{h}^{H}_{k}\mathbf{x}_{s,\tau}
	= \mathbf{h}^{H}_{k}\mathbf{h}^{\perp}_{\lambda\setminus \{s\} } X_{\chi} + \mathbf{h}^{H}_{k} \sum_{i=1}^{L-1}  \mathbf{h}^{\perp}_{\lambda\setminus\beta_{\chi,s}(i)} W^{d_{\beta_{\chi,s}(i)}}_{\tau}.\label{eqReceivedMessageAtCacheAidedUsers}
\end{align}

We can see that in~\eqref{eqReceivedMessageAtCacheAidedUsers}, the subfiles that are included in the summation term have all been cached by all receivers of set $\tau$, and as such they can be removed from the equation. What remains is XOR $X_{\chi}$ which, by design, is decodable by all users in $\tau$.

\begin{corollary}
	In Algorithm~\ref{algDeliveryXORplusUncoded}, each requested subfile $W_{\tau}^{d_{k}},k\in[K]$ is transmitted exactly $K\gamma+L$ times.
\end{corollary}

\begin{proof}
	Since each subfile ($W^{n}_{\tau}$) is divided into a total of $K\gamma+L$ smaller subfiles (cf. index $\phi$), in this section we aim to show that each of these subfiles will be transmitted exactly once i.e., that each $W_{\tau}^{d_{k}},k\in[K]$ appears in $K\gamma+L$ transmissions.
	
	We split the proof in two steps, where in the first we prove that any requested subfile will be transmitted $K\gamma+1$ times while combined with other subfiles to form the XOR message, while in the second we prove that any requested subfile is transmitted $L-1$ times while being a part of the uncoded messages.
	
	First, we can see that each XOR is transmitted a total of $K\gamma+1$ times, i.e. for each different $s\in\chi$ (Step $2$), which implies that requested subfile $W^{d_{k}}_{\tau}$ will be transmitted $K\gamma+1$ times as part of a XOR.
		
	Further, we can see that this subfile can be potentially transmitted, as part of the uncoded elements of the message, when $\chi$ is of the form $\chi=\tau\cup\{s\}$, while variable $s$ (Step $2$) takes values $s\in[K]\setminus\tau\setminus\{k\}$, thus $s$ can take $K-K\gamma-1$ different values.

	We can discern two cases, namely $K(1-\gamma)=L$ and $K(1-\gamma)>L$. In the first case where $ K-K\gamma=L$, it is clear that for every $s\in[K]\setminus\tau\setminus\{k\}$ then $\beta_{s,\tau}=[K] \setminus\tau\setminus\{s\}$, thus $k$ will be included in each transmission, which amounts to $L-1$ different subfiles.
	
	In the second case, where the size of set $[K]\setminus \tau$ is bigger than $L$, only a subset of the users will be selected every time to form sets $\beta_{s,\tau}, \forall s\in[K]\setminus\tau\setminus\{k\}$. Using $p_{k}\in\{1,2,..., K-K\gamma\}$, $k\in[K]\setminus\tau$ to denote the position of element $k$ in set $[K]\setminus\tau$, we can see that $k\in\beta_{s,\tau}$ if and only if for some $l\in\{2,...,L\}$, the following equality holds
\begin{align}\label{eqModuloConnection}
	p_{k}= ( p_{s} + l )\mod (K-K\gamma) -1.
\end{align}
The condition of~\eqref{eqModuloConnection} can be satisfied for exactly $L-1$ different values of $s$. Further, for any given $s$ it can only be satisfied by a single $l$, thus $k$ appears in exactly $L-1$ sets $\beta_{s,\tau}$, which completes the proof.\end{proof}

\section{Extension of the two type cache-aided scheme}\label{appExtensionTwoTypesNonIntegerStreams}

In this section we present an extension of Alg.~\ref{algDelivery2types} to accommodate any values $L_{1},L_{2}\in [1,L-1]$, such that $L_{1}+L_{2}=L$. The main premise is to increase the per-type subpacketization by some factor $d\in\mathbb{N}$, such that $d\cdot L_{1} \in\mathbb{N}, d\cdot L_{2} \in\mathbb{N}$. Hence the total subpacketization becomes
\begin{align}
	S=d^{2}(L_{1}+K_{1}\gamma_{1}) \binom{K_1}{K_{1}\gamma_{1}} (L_{2}+K_{2}\gamma_{2})\binom{K_{2}}{K_{2}\gamma_{2}}.
\end{align}

The new scheme works by repeating $d^2$ times Alg.~\ref{algDelivery2types}, with the difference that some transmissions will allocate $\lceil L_{1} \rceil$ streams to users of set $\mathcal{K}_{1}$ and at the same time will allocate $\lfloor L_{2} \rfloor$ streams to set $\mathcal{K}_{2}$ and in some transmissions will allocate $\lfloor L_{1} \rfloor$ streams to users of set $\mathcal{K}_{1}$ and at the same time will allocate $\lceil L_{2} \rceil$ streams to set $\mathcal{K}_{2}$.

This way, it allows the average allocation of $L_{1}$ and $L_{2}$ streams to each user type, which leads to the DoF $D_{L}(K_{1},\gamma_{1},K_{2},\gamma_{2})=L+K_{1}\gamma_{1}+K_{2}\gamma_{2}$.

{

\begin{example}	
	Let us consider the setting with parameters $K_{1} = K_{2} = 10$, $\gamma_{1} = \frac{2}{10}$, $\gamma_{2} = \frac{1}{10}$, and $L=6$. The optimal allocation of the spatial multiplexing resources is $L_1 = \frac{38}{17}$ and $L_2 = \frac{64}{17}$, according to \eqref{eqStreamsCalculation}.
	
	To accommodate for the non-integer $L_1, L_2$ we increase the subpacketization by a factor of $17^2$. By repeating Algorithm~\ref{algDelivery2types} a total of $17\times 13$ iterations with values $L_1' = 2$ and $L_{2}'=4$ and subsequently repeating Algorithm~\ref{algDelivery2types} a total of $17\times 4$ iterations with values $L_1'' = 3$ and $L_{2}''=3$ we can serve the demands of all users with a constant DoF of $K_{1}\gamma_1+K_{2}\gamma_{2}+L = 9$.

\end{example}

}

\bibliographystyle{ieeetr}

\enlargethispage{-1.2cm}

\end{document}